\documentclass[showpacs,nofootinbib,showkeys,eqsecnum,prd,aps,preprint]{revtex4-1}

\usepackage[english]{babel}
\usepackage[cp1251]{inputenc}
\usepackage[T1]{fontenc}

\usepackage{amssymb}
\usepackage{amsmath}
\usepackage{amsfonts}
\usepackage{latexsym}
\usepackage{bm}
\usepackage{tensor}
\usepackage{hyperref}

\newtheorem{lemma}{Lemma}
\newtheorem{definition}{Definifion}
\newtheorem{theorem}{Theorem}

\newenvironment{proof}[1][Proof]
{\textbf{#1} }{\ \rule{0.5em}{0.5em}}

\begin{document}


\title{Symmetry operators  and separation of variables in the $(2+1)$-dimensional Dirac equation with external electromagnetic field}

%
%

\title{Symmetry operators  and separation of variables in the $(2+1)$-dimensional Dirac equation with external electromagnetic field
}

\author{A.~V.~Shapovalov${}^{a,b}$}
\email{shpv@phys.tsu.ru}

\author{A.~I.~Breev${}^{a}$}
\email{breev@mail.tsu.ru}

\date{\today }

%
%

%
%




\begin{abstract}
We obtain and analyze  equations determining first-order differential symmetry operators  with matrix coefficients for the Dirac equation with an external electromagnetic potential  in a $(2+1)$-dimensional Riemann (curved) spacetime. Nonequivalent complete sets of mutually commuting symmetry operators are classified in a $(2+1)$-dimensional Minkowski (flat) space.  For each of the sets we carry out a complete separation of variables in the Dirac equation and find a corresponding electromagnetic potential permitting separation of variables.
\end{abstract}


\affiliation{
$^{a}$Department of Theoretical Physics, Tomsk State University, Novosobornaya Sq.1, 634050, Tomsk, Russia\\
${}^{b}$Department of Higher Mathematics and Mathematical Physics, \\
Tomsk Polytechnic University, Lenin Ave., 30, 634034, Tomsk, Russia
}
\maketitle

\section{Introduction}
\label{introduction}

In this paper, we consider a complete separation of variables for the  Dirac equation in a $(2+1)$-dimensional spacetime  (the $(2+1)$-dimensional Dirac equation) with an external electromagnetic field in a Minkowski spacetime with the use of mutually commuting first-order differential symmetry operators with matrix coefficients.
This problem was solved for the Dirac equation in the usual $(3+1)$-dimensional  Minkowski spacetime  in the early 1970s
\cite{vnsh_ecl1973, bagr_vnsh_ecl1975}  on the basis of the theory of separation of variables in linear systems of first-order partial differential equations  (PDEs) and in second-order scalar PDEs  \cite{bagr_vnsh_ecl1975,bagr_m_sh_vnsh1973,vnsh1978}.

This advancement in theory gave rise to a systematic search for new exact solutions to  the Dirac and Klein--Gordon equations with external electromagnetic fields and to a classification of the fields permitting separation of variables \cite{bagr_m_sh_vnsh1973}.

A comprehensive review and detailed analysis of the known solutions obtained by separation of variables was carried out by Bagrov and Gitman in the book \cite{bagr_git_book2014}.

The Dirac equation in $(2+1)$ dimensions has been actively studied for the past three decades by many researchers. We do not intend here
to provide an exhaustive literature review related to physical problems where solutions of the $(2+1)$-dimensional  Dirac equation are essential and refer only to some articles to illustrate a variety of contexts where these solutions can be applied.
The topical problem in relativistic quantum theory is the self-adjoint extension of the Dirac Hamiltonian in external singular potentials.
The Dirac equation for a  magnetic solenoid field is the basis of the theory of the Aharonov--Bohm effect both in $(3+1)$ and in $(2+1)$ dimensions
\cite{gerb1989,bagr_git_tl2001,gavr_git_sm2004}.
Khalilov \cite{khal2013} considered  the Dirac equation in $(2+1)$ dimensions for a relativistic charged zero-mass fermion
in Coulomb and Aharonov-Bohm potentials in the  context of a self-adjoint extension problem.
The self-adjoint extension problem in quantum mechanics was investigated in detail by Gitman, et al. \cite{git_tyut_voron2012}.

Exact solutions of the $(2+1)$-dimensional  Dirac equation for an external  field in the presence of a minimal length were studied by Menculini et al. \cite{roy2013,roy2015}.

The Dirac equation is of interest in studying planar gravity and BTZ black holes \cite{BTZ1992,BTZ1993, BTZ1998} and also in cosmology \cite{Brevik, Nj2003}.

 Another motivation in studying the $(2+1)$-dimensional  Dirac equation in curved spacetime is that,
 although the $(2+1)$-dimensional gravity is a toy model for a regular Einstein theory in $(3+1)$ dimensions,
 it preserves some significant properties of regular gravity being mathematically simpler
 (see, e.g., \cite{SucuUnal2007}).

Finally, it should be noted that in condensed matter physics, the $(2+1)$-dimensional Dirac equation with an external electromagnetic potential is used to study theoretically  the electronic properties
of graphene (see \cite{Castro2009,Katsnelson2010,Katsnelson2012}), graphene quantum dots \cite{Gulcu2014}, etc.

Summarizing, we can see that the  development of mathematical techniques
to find new exact solutions for the $(2+1)$-dimensional Dirac equation is important in view of extending
applications of relativistic quantum equations to topical problems of quantum theory.

\par

The paper is organized as follows.  In Section 2, basic definitions regarding the $(2+1)$-dimensional Dirac equation with an external electromagnetic potential in pseudo-Riemannian spacetime are given.
In Section 3, we obtain and  analyze in detail the equations determining a first-order symmetry operator with matrix coefficients for the Dirac equation in the $(2+1)$-dimensional Riemann spacetime.
In Section 4, the focus is on classifying complete sets of symmetry operators in $(2+1)$-dimensional Minkowski space.
Separation of variables in the $(2+1)$-dimensional Dirac equation is carried out in Section 5 with the use of the complete sets of symmetry operators in $(2+1)$-dimensional Minkowski space. The separable coordinates and the corresponding electromagnetic potentials permitting separation of variables are presented.

\section{The Dirac equation in  $(2+1)$-dimensional Riemann spacetime}
\label{sect1}

\subsection{ A $(2+1)$-dimensional Riemann spacetime}
\label{sect1-1}
We introduce the notation used in the Dirac equation in a curved spacetime that is described as a $(2+1)$-dimensional pseudo-Riemannian manifold $\mathcal{M}(g)$. A pseudo-Riemannian metric $(g_{\mu\nu})$ on the manifold $\mathcal{M}$
is given by its covariant components $g_{\mu\nu}(x)$  in  local coordinates $x=(x^\mu)=(x^0,x^1,x^2)$ on $\mathcal{M}$;  $\mu,\nu,\alpha, \dots  =0,1,2$. The contravariant components of the metric tensor  are $g^{\mu\nu}$, $g^{\mu\nu}g_{\nu\alpha}=\delta^\mu_\alpha $ \footnote{The Einstein rule of summation is used here and below.}. Here $\delta^\mu_\alpha$ is the the Kronecker delta ($\delta^\mu_\mu =1$ and 0 otherwise).

The Levi-Civita connection on the tangent bundle $T\mathcal{M}$ is described by the Christoffel symbols
\begin{equation}
	\label{levi-chivita1}
	\Gamma^\mu_{\nu\alpha}=\frac{1}{2}g^{\mu\beta}\left(
		\partial_\alpha g_{\beta\nu} +
		\partial_\nu g_{\beta\alpha} -
		\partial_\beta g_{\nu\alpha} \right),\quad
		\partial_\mu = \frac{\partial}{\partial x^\mu}.
\end{equation}
The Riemann curvature tensor is  given in terms of the Christoffel symbols as
\begin{eqnarray}
	\label{riemanntensor}
	-R^\sigma_{. \alpha \nu \mu} =
	\partial_\nu \Gamma^\sigma_{\alpha \mu } -
	\partial_\mu \Gamma^\sigma_{\alpha \nu } +
	\Gamma^\sigma_{\beta \nu }\Gamma^\beta_{\alpha \mu } -
	\Gamma^\sigma_{\beta \mu }\Gamma^\beta_{\alpha \nu }
\end{eqnarray}
(we follow here the notation of \cite{DubrNovFom1992}). The  Ricci tensor $R_{\mu\nu}$ and curvature $R$ are
\begin{eqnarray}
	\label{riccitensor}
	R_{\nu\mu}=R^\alpha_{. \nu \alpha \mu} ,\quad R = R^\mu_\mu.
\end{eqnarray}

The curvature tensor (\ref{riemanntensor}) in the three-dimensional manifold $\mathcal{M}$
is uniquely determined through the metric and the Ricci tensor (\ref{riccitensor}) as
\begin{gather}\nonumber
	R_{\nu\mu\alpha\beta} = R_{\nu\alpha}g_{\mu\beta} - R_{\nu\beta}g_{\mu\alpha} +
				R_{\mu\beta}g_{\nu\alpha} - R_{\mu\alpha}g_{\nu\beta} +
				\frac{R}{2}\left( g_{\nu\beta}g_{\mu\alpha} - g_{\nu\alpha}g_{\mu\beta} \right).
\end{gather}
Covariant differentiation of a contravariant vector $V^\mu$ and a covariant vector $V_\mu$ yields
$\nabla_\nu V^\mu=V^\mu_{;\nu}=V^\mu,_\nu+\Gamma^\mu_{\nu\alpha}V^\alpha $ and $\nabla_\nu V_\mu=V_{\mu;\nu}=V_{\mu},_{\nu}-\Gamma^\alpha_{\nu\mu}V_\alpha $, respectively. Here the semicolon denotes the covariant derivative with respect to coordinate indices, and the comma implies a partial derivative.
The connection $\Gamma^\mu_{\nu\alpha}$ is metric compatible, i.e. $g_{\alpha\beta;\mu}=0$.

To define spin connection on the $(2+1)$-dimensional spacetime, we need  local frame fields $e_\mu^a(x)$
that diagonalize the metric:
 \begin{eqnarray}
\label{triada1}
g_{\mu\nu}(x)=e_\mu^a(x)e_\nu^b(x)\eta_{ab},
\end{eqnarray}
where $\eta_{ab}$ is the Minkowski metric, $(\eta_{ab})={\mbox{diag} (1,-1,-1) }$. Here, Latin letters denote the local
frame indices, $a,b=0,1,2$. Also, $e^{\mu a}(x)=g^{\mu\nu} e_{\nu}^a(x)$, $e^{\mu}_ a(x)=g^{\mu\nu} e_{\nu}^b(x)\eta_{ab}$,
$\eta^{ac}\eta_{cb}=\delta^a_b$.

In $(3+1)$-general relativity, a frame field  is also called a tetrad or a vierbein field;  in $(2+1)$ spacetime we use the term triad for $e_\mu^a(x)$, and then $a,b, c $ are triad indices.

In what follows, we will need the Levi--Civita antisymmetric tensor $e_{\mu\nu\alpha}(x)$ on the $(2+1)$-dimensional manifold $(\mathcal{M},g)$, which is defined as
 \begin{equation}
	\label{l-c-tensor1}
	e_{\mu\nu\alpha}(x) = \operatorname{det}\left( e^a_\mu(x) \right) \epsilon_{\mu\nu\alpha}.
\end{equation}
By $\epsilon_{\mu\nu\alpha}$ we denote the Levi--Civita antisymmetric symbol, which is defined as $\epsilon_{012}=1$.
We also note that
 \begin{eqnarray}
	\label{l-c-tensor3}
	e^{\alpha\beta\gamma}e_{\alpha\beta\gamma} = 6,\quad
	e^{\alpha\beta\gamma}e_{\mu\beta\gamma} = 2 \delta^\alpha_\mu,\quad
	e^{\alpha\mu\nu}e_{\alpha\beta\gamma} = \delta^\mu_\beta \delta^\nu_\gamma - \delta^\mu_\gamma \delta^\nu_\beta.
\end{eqnarray}

\subsection{ The Dirac equation}
\label{sect1-2}

The Dirac equation in $(2 + 1)$-dimensional spacetime is determined by introducing the corresponding Dirac matrices, the spin connection, and the generalized momentum operator.

\subsubsection*{The Dirac matrices}

We follow the notation of \cite{bagr_git_book2014} for $(2+1)$  $\gamma$- matrices:
\begin{eqnarray}
\label{gammamatr1}
\gamma^\mu(x)=e^\mu_a(x)\hat \gamma^a,
\end{eqnarray}
where
\begin{eqnarray}
\label{gammamatr2}
\hat\gamma^0=\sigma_3, \quad \hat\gamma^1=i s \sigma_1, \quad \hat\gamma^2=\sigma_2,
\end{eqnarray}
and $\sigma_1,\sigma_2,\sigma_3$ are the Pauli spin matrices \footnote{
The Pauli matrices
\begin{eqnarray*}
\label{gammamatr3}
\sigma_1=
\begin{pmatrix}
  0 & 1 \\
  1 & 0 \\
\end{pmatrix},
\quad
\sigma_2=
\begin{pmatrix}
  0 & -i \\
  i & 0 \\
\end{pmatrix},
\quad
\sigma_3=
\begin{pmatrix}
  1 & 0 \\
  0 & -1 \\
\end{pmatrix}
\end{eqnarray*}
are Hermitian, traceless,  and possessing the properties:
\begin{gather}\nonumber
	\sigma_1 \sigma_2 =- \sigma_2\sigma_1 = i  \sigma_3,\quad
	\quad
	\sigma_2 \sigma_3 =- \sigma_3\sigma_2 = i\sigma_1,\quad
	\quad
	\sigma_3 \sigma_1 =- \sigma_1\sigma_3 = i\sigma_2.
\end{gather}
}.
 The four matrices
 \begin{eqnarray}
\label{gammamatr5}
\gamma^\mu(x), \,\, I
\end{eqnarray}
 form a basis for the set of $2\times 2$ matrices, where $I$ is the identity matrix.

The algebraic properties of  $(2+1)$ gamma-matrices (\ref{gammamatr1}), (\ref{gammamatr2})  are
\begin{equation}
	\label{gammamatr6-1}
	\gamma^\alpha\gamma^\beta+\gamma^\beta\gamma^\alpha=2 g^{\alpha\beta}, \quad
	\gamma^\alpha \gamma^\mu = g^{\alpha \mu}-i s e^{\alpha\mu\sigma}\gamma_\sigma,\quad
	[\gamma^\mu, \gamma^\nu] = - 2 i s e^{\mu\nu\sigma}\gamma_\sigma.
\end{equation}
where $[A,B]=AB-BA $ is the commutator of $A$ and $B$.

\subsubsection*{The spinor connection}

The spinor connection $\Gamma_\mu$  provides covariant differentiation of  spinors: if $\psi$ is a spinor, then
$\nabla_\mu\psi+\Gamma_\mu\psi=\psi_{;\mu}+\Gamma_\mu\psi$ is also a spinor. The matrices $\Gamma_\mu$ are assumed to be traceless:
\begin{equation}
\label{spinconn1}
\mbox{tr}(\Gamma_\mu)=0.
\end{equation}
For the spinor-covariant derivative of the Dirac matrices $\gamma^\mu$  we have
\begin{equation}
\label{spinconn2}
[\nabla_\alpha+\Gamma_\alpha,\gamma^\mu]=0 \quad \mbox{or} \quad  \gamma^\mu_{;\alpha}=-[\Gamma_\alpha , \gamma^\mu].
\end{equation}
It is easy to verify that
\begin{equation}
\label{spinconn3}
\Gamma_\nu=-\displaystyle\frac{1}{4}\gamma^\alpha_{;\nu}\gamma_\alpha.
\end{equation}
Indeed, since $\Gamma_\nu$ are traceless (\ref{spinconn1}),
we can expand $\Gamma_\nu$ in terms of the basis (\ref{gammamatr5}) as $\Gamma_\nu=a_\nu^\alpha\gamma_\alpha$
with  constant expansion coefficients $a_\nu^\alpha$. From (\ref{spinconn2}) it follows that
$\gamma^\mu_{;\nu}+a^\alpha_\nu(\gamma_\alpha\gamma^\mu-\gamma^\mu\gamma_\alpha)=0$. Multiplying this equation by  $\gamma_\mu$ and
summing over $\mu$, we get
$\gamma^\mu_{;\nu}\gamma_\mu+4 a^\alpha_\nu\gamma_\alpha=0$ and consequently (\ref{spinconn3}).

\subsubsection*{ The generalized momentum operator}

The components $\mathcal{P}_\mu$ of the generalized momentum operator are defined as
\begin{equation}
\label{momentum1}
\mathcal{P}_\nu=p_\nu-A_\nu , \quad p_\nu=i(\nabla_\nu+\Gamma_\nu),
\end{equation}
where $A_\nu$ are the components of the vector potential of the  external electromagnetic field. The components of the electromagnetic tensor are $F_{\nu\mu}=\nabla_\nu A_\mu - \nabla_\mu A_\nu=A_{\mu, \nu}-A_{\nu, \mu}$.

\begin{theorem}
	The commutators of $\mathcal{P}_\nu$ read
	 \begin{equation}\nonumber	
[\mathcal{P}_\nu,\mathcal{P}_\mu] = -[\nabla_\nu ,\nabla_\mu] -\displaystyle\frac{i}{4}\, s \, R_{\nu\mu\, \, . \,\, .\,\,}^{\, . \,\, . \,\,\,\  \alpha\beta} e_{\alpha\beta\sigma}\gamma^\sigma-iF_{\nu\mu}.
\end{equation}
	\label{}
\end{theorem}
\begin{proof}
	\begin{equation}
\label{elmag2}
[\mathcal{P}_\nu,\mathcal{P}_\mu] = -[\nabla_\nu ,\nabla_\mu] -(\Gamma_{\mu ;\nu}-\Gamma_{\nu ;\mu} ) -[\Gamma_\nu, \Gamma_\mu ]-iF_{\nu\mu},
\end{equation}
where
\begin{equation*}
\Gamma_{\mu ;\nu}-\Gamma_{\nu ;\mu} = -\displaystyle\frac{1}{4}\gamma^\alpha R^\beta_{.\,\alpha \nu\mu}.
\end{equation*}
Now we define an operator $\hat G_{\nu\mu}=-(\Gamma_{\mu ;\nu}-\Gamma_{\nu ;\mu}+ [\Gamma_\nu,\Gamma_\mu])$ and rewrite equation (\ref{elmag2}) as
\begin{equation}
\label{elmag3}
[\mathcal{P}_\nu,\mathcal{P}_\mu] = -[\nabla_\nu ,\nabla_\mu] +\hat G_{\nu\mu}-iF_{\nu\mu}.
\end{equation}
Then it is easy to see that the tracelesness property of $\Gamma_\mu$ in (\ref{spinconn1})
immediately leads to
\begin{equation}
\label{tracegbig1}
\mbox{tr}(\hat G_{\mu\nu})=0.
\end{equation}
Note that  $[[\mathcal{P}_\nu,\mathcal{P}_\mu],\gamma^\alpha]=0$ is the obvious consequence of (\ref{spinconn2}). Hence, we have
 \begin{equation}
\label{gammabig3}
-[[\nabla_\nu,\nabla_\mu],\gamma^\alpha]+[\hat G_{\nu\mu},\gamma^\alpha]=0.
\end{equation}
 On the other hand,
\begin{equation*}
[[\nabla_\nu,\nabla_\mu],\gamma^\alpha]=\gamma^\alpha_{;\mu\nu}-\gamma^\alpha_{;\nu\mu}=
R^{\beta\alpha}_{.\, .\, \nu\mu}\gamma_\beta.
\end{equation*}
 Then (\ref{gammabig3}) yields
 \begin{equation}
\label{gammabig5}
R_{\mu\nu\beta\, .}^{\,\,\,\,\,\,\,\ \alpha }\gamma^\beta +[\hat G_{\nu\mu}, \gamma^\alpha]=0.
\end{equation}
Substituting in (\ref{gammabig5})  the expansion formula
 \begin{equation}
\label{gammabig6}
\hat G_{\nu\mu}=a_{\nu\mu\beta}\gamma^\beta
\end{equation}
for $\hat G_{\nu\mu}$ in terms of the basis (\ref{gammamatr5}) with the use of (\ref{tracegbig1}), we get
 \begin{equation}
\label{gammabig7}
R_{\mu\nu\beta\, .}^{\,\,\,\,\,\,\,\ \alpha }\gamma^\beta +a_{\nu\mu\beta}[\gamma^\beta,\gamma^\alpha]=0.
\end{equation}
Then, applying the commutation relation (\ref{gammamatr6-1}) to (\ref{gammabig7}), we have
  \begin{equation}
\label{gammabig8}
R_{\mu\nu\beta\, .}^{\,\,\,\,\,\,\,\ \alpha }\gamma^\beta - 2is\, a_{\nu\mu\beta}\, e^{\beta\alpha\sigma} \gamma_\sigma=0 ,
\end{equation}
and using (\ref{l-c-tensor3}) we find that
\begin{equation*}
a_{\nu\mu\alpha}=-\displaystyle\frac{i}{4}\, s \, e_{\alpha \sigma\beta} R^{\sigma\beta\, \, . \,\, .\,\,}_{\, . \,\, . \,\,\,\  \nu\mu } \, .
\end{equation*}

Finally, substituting (\ref{gammabig8}) and (\ref{gammabig6}) in equation (\ref{elmag3}), we obtain (\ref{elmag2}).
\end{proof}

Now we can write the Dirac equation for a Dirac particle of mass $m$
in the
 $(2+1)$-dimensional pseudo-Riemannian manifold  $\mathcal{M}$ with external electromagnetic potential $A_\nu$ as follows:
  \begin{equation}
\label{dirac1}
H\psi=0, \quad H=\gamma^\nu\mathcal{P}_\nu-m .
\end{equation}

\section{The determining equations}
\label{sect2}

Separation of variables in the Dirac equation (\ref{dirac1}) involves a matrix-valued first-order differential symmetry operator $X$,
\begin{equation}
	\label{symm-oper1}
	X=X^{\nu}\mathcal{P}_\nu+\chi ,
\end{equation}
mapping each solution of (\ref{dirac1}) in another  solution  of this equation. Here  $X^{\nu}$ and $\chi$ are matrix functions of $x$.

The determining equation for the symmetry operator $X$ can be written as
\begin{equation}
	\label{determ-eq1}
	[X,H]=\Psi H.
\end{equation}
Here $\Psi$ is a Lagrange operator multiplier having the form
\begin{equation}
	\label{determ-eq1-1}
	\Psi=\Psi^\nu \mathcal{P}_\nu + \bar\Psi ,
\end{equation}
where $\Psi^\nu$ and $\bar\Psi $ are matrix functions of $x$.

The left and right sides of  equation (\ref{determ-eq1}) act on smooth scalar functions of $x$ from the general domain of definition of the operators $X$ and $H$.

Operators (\ref{symm-oper1}) that satisfy equations (\ref{determ-eq1}) form a Lie algebra $\mathfrak{g}$. It is clear, that any operator $Y = R  H$ is a solution of (\ref{determ-eq1}) with $\Psi = [R,H]$, where $R$ is a linear differential operator. The set of operators $Y$ forms an ideal $\mathfrak{h}$ in the Lie algebra $\mathfrak{g}$:
\begin{equation}\nonumber
	[Y, X] = \left( [R,X] - R \Psi \right) H.
	\label{}
\end{equation}
The operators $Y$ belonging to  $\mathfrak {h}$ do not contain any information about solutions of equation (\ref{dirac1}). These symmetry operators will be called trivial. We are interested in the elements of the quotient algebra $\mathfrak {l}=\mathfrak {g}/\mathfrak {h}$, which we call non-trivial symmetry operators.

Substituting  expressions (\ref{symm-oper1}) and (\ref{determ-eq1-1}) in (\ref{determ-eq1}) and equating
coefficients of the same  powers of $\mathcal{P}_\mu$, we come to the following result:
\begin{lemma}\label{lemm01}
	 The determining equation (\ref{determ-eq1}) is equivalent to the following system of equations for $X^{\nu}$, $\chi$ and $\Psi^\mu$, $\bar\Psi$:
	\begin{gather}
		 \label{commut-deter-eq3}
		 [\gamma^\mu, X^\nu]+[\gamma^\nu, X^\mu] =\Psi^\nu\gamma^\mu+\Psi^\mu\gamma^\nu, \\
		 \label{commut-deter-eq4}
		 [\gamma^\mu, \chi]+\gamma^\nu [\mathcal{P}_\nu, X^\mu] =\bar\Psi \gamma^\mu + m\Psi^\mu, \\
		 \label{commut-deter-eq5}
		 \gamma^\nu[\mathcal{P}_\nu, \chi]+\displaystyle\frac{1}{2}\{\gamma^\nu, X^\mu \}_{+}[\mathcal{P}_\nu, \mathcal{P}_\mu]=
		\displaystyle\frac{1}{2}\Psi^\mu\gamma^\nu[\mathcal{P}_\nu, \mathcal{P}_\mu] + m\bar\Psi.
\end{gather}
\end{lemma}

Note that Lemma \ref{lemm01} is true not only for the $(2+1)$ Dirac equation
(\ref{dirac1})  but also for the $(3+1)$ one.

Expand the matrix functions $X^{\nu}$, $\chi$, $\Psi^\nu$, $\bar\Psi$ in terms of the basis (\ref{gammamatr5}):
\begin{gather}
	\label{symm-oper2}
	X^{\nu}=\xi^\nu +\gamma^\alpha X^{\nu}_{.\,\alpha} , \quad \chi= \varphi+\varphi^\alpha \gamma_\alpha,\\
	\label{PsiTwo}
	\Psi^\nu = \Psi^{\nu}_{.\,\alpha}\gamma^\alpha+\bar\Psi^\nu,\quad	
	\bar\Psi = \bar\Psi^{\nu}_{0}\, \gamma_\nu+\bar\Psi_{0},
\end{gather}
where
\begin{equation}
	\xi^\nu, \quad X^{\nu}_{.\,\alpha},\quad  \varphi,\quad  \varphi^\alpha,\quad \Psi^{\nu}_{.\,\alpha},\quad \bar\Psi^\nu,\quad \bar\Psi^{\nu}_{0},\quad \bar\Psi_{0}
	\label{FullFunDer}
\end{equation}
are smooth scalar functions of $x$.

Equations (\ref{commut-deter-eq3}) --- (\ref{commut-deter-eq5}) for the matrix coefficients $X^{\nu}$, $\chi$ and  $\Psi^\nu$, $\bar\Psi$ of the  symmetry operator (\ref{symm-oper1}) and of the Lagrange multiplier, respectively,
 lead to  the corresponding equations for scalar functions  (\ref{FullFunDer})
which are given below in terms of  the following lemmas.
\begin{lemma}\label{lemm02}
	Equation (\ref{commut-deter-eq3}) gives
	\begin{gather}
		\label{de01}
		\Psi^{\mu\nu} = 2 X^{\mu\nu} - \frac{2}{3}\operatorname{Sp}(X)g^{\mu\nu},\quad		
		\bar\Psi^{\nu}=0,\\
		\label{de02x}
		X^{\mu\nu} + X^{\nu\mu} = \frac{2}{3}\operatorname{Sp}(X)g^{\mu\nu},\quad
		\operatorname{Sp}(X) = X^\mu_{\cdot\mu}.
	\end{gather}
\end{lemma}
\begin{proof}
Substituting (\ref{symm-oper2}) and (\ref{PsiTwo}) in equation (\ref{commut-deter-eq3}) and  expanding the latter in terms of the basis (\ref{gammamatr5}) with the use of formulas (\ref{gammamatr6-1}), we get
	\begin{gather}\label{proof01}
		- i s \left(
		e^{\alpha\nu\sigma}(2 X^{\mu}_{\cdot\alpha}-\Psi^\mu_{\cdot\alpha}) +
 	        e^{\alpha\mu\sigma}(2 X^{\nu}_{\cdot\alpha}-\Psi^\nu_{\cdot\alpha})
		      \right)\gamma_\sigma = \\ \nonumber
		= \Psi^{\mu\nu} + \Psi^{\nu\mu} + \overline{\Psi}^\mu\gamma^\nu + \overline{\Psi}^\nu\gamma^\mu.
	\end{gather}
	 From  (\ref{proof01}), we immediately obtain $\Psi^{\mu\nu}=-\Psi^{\nu\mu}$, and setting $\nu=\mu$ in  (\ref{proof01}), we get
	\begin{equation*}
		- i s e^{\alpha\mu\sigma}(2 X^{\mu}_{\cdot\alpha} - \Psi^{\mu}_{\cdot\alpha}) = \overline{\Psi}^{\mu}g^{\mu\sigma}.
	\end{equation*}
For $\alpha=\sigma$ we have $\overline{\Psi}^\mu = 0$, and for $\mu \neq \alpha$
it follows that

\begin{equation}
		\Psi^\mu_{\cdot\alpha} = 2 X^\mu_{\cdot\alpha}.
		\label{proof03}
	\end{equation}
For arbitrary indices $\mu$ and $\alpha$, we can write (\ref{proof03}) as (\ref{de01}). Substituting  (\ref{de01}) into the condition $\Psi^{\mu\nu}=-\Psi^{\nu\mu}$, we  obtain equation (\ref{de02x}).
\end{proof}

Taking Lemma \ref{lemm02} into account, we find that symmetry operator
(\ref{symm-oper2}) takes the form
\begin{equation*}
	X = \left( \xi^\alpha + \frac{1}{2}\Psi^\alpha{}_\beta \gamma^\beta \right)\hat{\mathcal{P}}_\alpha +
	\frac{1}{3}\operatorname{Sp}(X)\hat H + \left( \chi + \frac{1}{3}\operatorname{Sp}(X)m \right).
	\label{}
\end{equation*}
Since we are interested in nontrivial symmetry operators, we can factorize them  into  elements of the ideal $\mathfrak{h}$. Then,
without loss of generality, we can set
\begin{equation}
	\operatorname{Sp}(X) = 0.
	\label{SpXcond}
\end{equation}
Next, in view of (\ref{SpXcond}), we can write  the symmetry operator in the form
\begin{equation}\nonumber
	X = \left( \xi^\alpha - \frac{i s}{2}X_{\mu\nu}e^{\mu\nu\alpha} \right)\hat{\mathcal{P}}_\alpha +
	\left(\frac{i s}{2}X^{\mu\nu}e_{\mu\nu\alpha}\gamma^\alpha\right)\hat H +
	\left[\chi + m \frac{i s}{2}X^{\mu\nu}e_{\mu\nu\alpha}\gamma^\alpha\right].
\end{equation}
Thereby, without loss of generality, we can explore only those symmetry operators (\ref{symm-oper1}) whose coefficients of the derivatives do not include any other matrix except the identity matrix. In other words, we can set in (\ref{symm-oper2}): $X_{\mu\nu} = 0$.

\begin{lemma}\label{lemm03}
	Equation (\ref{commut-deter-eq4}) provides
	\begin{gather}
		\label{de031}
		\overline{\Psi}_0 = -\frac{i}{3}\xi^{\mu}{}_{;\mu},\quad \Psi_{0\mu} = 0,\\
		\label{de032}		
		\varphi_\alpha = -\frac{s}{4}\xi^{\mu;\sigma} e_{\mu\sigma\alpha},\\
		\label{de033}
		 \xi^{\mu;\sigma} + \xi^{\sigma;\mu} = \frac{2}{3} \xi^\alpha{}_{;\alpha} g^{\mu\sigma}.
	\end{gather}	
\end{lemma}
\begin{proof}
Substitute (\ref{symm-oper2}) and (\ref{PsiTwo}) in equation (\ref{commut-deter-eq4}) and expand it in terms of the basis (\ref{gammamatr5}). Then we obtain the system of equations
	\begin{gather}\label{pr032}
		{\Psi_0}^\mu = 0,\quad		
		i s e^{\alpha\mu\sigma}  (2\varphi_\alpha - \Psi_{0\alpha}) = - i \xi^{\mu;\sigma} - \overline{\Psi}_0 g^{\mu\sigma}.
	\end{gather}
 Contraction of equation (\ref{pr032}) with $g_{\mu\sigma}$ gives (\ref{de031}), and its contraction with $e_{\tau\mu\sigma}$
 results in (\ref{de032}). From (\ref{de031}), (\ref{de032}) and (\ref{pr032}), in view of (\ref{SpXcond}), we get (\ref{de033}).
\end{proof}

Substituting (\ref{symm-oper2}) and (\ref{PsiTwo}) in  equation (\ref{commut-deter-eq5}), and taking into account Lemmas \ref{lemm02} and \ref{lemm03} and  conditions (\ref{SpXcond}), we obtain the following lemma:
\begin{lemma}\label{lemm04}
From equation (\ref{commut-deter-eq5}) it follows that
	\begin{gather}
		\label{de0410}
		\varphi_{,\beta} =
		\frac{i}{3} \xi^\nu{}_{;\nu\beta} + F_{\beta\mu} \xi^\mu,\\
		\label{de0420}
		m \xi^\mu{}_{;\mu} = 0.
	\end{gather}	
\end{lemma}

Let us summarize the results obtained in the form of a theorem.

\begin{theorem}
The symmetry operator (\ref{symm-oper1}) of the Dirac equation (\ref{dirac1}) is of the form
	\begin{equation}
		\hat X = \xi^\mu\hat{\mathcal{P}}_\mu - \frac{s}{4}\xi^{\mu;\nu}e_{\mu\nu\alpha}\gamma^\alpha + \varphi,
		\label{opXsymm}
	\end{equation}
where the vector field $\xi^\mu$ is determined by the system of equations (\ref{de033}) and (\ref{de0420}). The function  $\varphi$ is found from equation (\ref{de0410}).
\end{theorem}

It is natural to consider the symmetry operators $X$ as analytical functions of the real mass parameter $m$ entering into the Dirac equation in a neighborhood of the point $m=0$.
Then condition  (\ref{de0420}) leads to exploration separately  the massive case ($m\neq 0$) and the massless one ($m=0$).

In the present work, we deal with the massive case. Then from Lemma  \ref{lemm04} it follows that
\begin{equation}
	\hat \Psi = -i \xi^\mu{}_{;\mu} = 0.
	\label{PsiLag}
\end{equation}
We find that $\xi^\mu$ is a Killing vector field:
\begin{equation}
	 \label{killingEquation}
	 \xi_{\nu;\mu}+\xi_{\mu;\nu}=0.
\end{equation}

Note that the symmetry operators $X$ with $\hat \Psi = 0$ in (\ref{determ-eq1}) were used for separation of variables in the  $(3+1)$-dimensional Dirac equation in  Minkowski spacetime \cite{vnsh_ecl1973} (see also \cite{bagr_git_book2014}).

In other words, for the $(2+1)$-dimensional  Dirac equation, the first-order symmetry operator $X$  commuting with the Dirac operator $H$ of the form (\ref{dirac1}) has only scalar coefficients of  the derivatives, which are  Killing vector fields.
This situation is different from the case of the $(3+1)$-dimensional  Dirac equation when the symmetry operator has
both scalar and matrix coefficients of the first derivatives.
Here, by a scalar coefficient we mean a scalar function multiplied by the identity matrix.

\section{Minkowski spacetime}
\label{minkowskii}

Consider symmetry operators (\ref{opXsymm}) in $(2+1)$-dimensional Minkowski spacetime with the metric $(g_{\mu\nu})={\rm diag (1,-1,-1) }$ in Cartesian coordinates $(x)=(x^\mu)$.

The solution of the Killing equations in a plane space ($R^{\alpha \, .\, .\, .\,}_{\, .\, \nu\mu \beta }=0$) is well known. Using the designation $S=\xi_\nu a^\nu$, where $a^\nu$ are arbitrary real parameters, we have from (\ref{killingEquation}): $S_{;\nu} a^\nu=0$. Solving this equation with the use of characteristics, we find $S= 2 A^{\alpha\nu}x_\alpha a_\nu +   B^\nu a_\nu $,
 where $ A^{\alpha\nu} ( =- A^{\nu\alpha})$ and $ B^\nu$ are arbitrary real constants of integration, $x_\mu=g_{\mu\nu}x^\nu$. Then  the Killing vector field $\xi^\nu$  reads (\ref{symm-oper2}):
 \begin{equation}
	\label{killing-1}
	\xi^\nu= 2 A^{\alpha\nu}x_\alpha +  B^\nu.
\end{equation}
In Cartesian coordinates, the spinor connection $\Gamma_\mu =0$ in (\ref{spinconn3}), and therefore the momentum operator (\ref{momentum1}), takes the form $\mathcal{P}_\nu=p_\nu-A_\nu , \quad p_\nu=i\partial_\nu$.

As a result,  we can write (\ref{opXsymm}) as
\begin{equation}
	\label{symm-oper-zero-1-1}
	X = X(x,\mathcal{P})=A^{\alpha\nu}L_{\alpha\nu}+B^\nu \mathcal{P}_\nu + \varphi + \varphi_\alpha\gamma^\alpha ,
\end{equation}
where we use the designation $L_{\alpha\nu}=x_\alpha\mathcal{P}_\nu - x_\nu \mathcal{P}_\alpha$.
Equation (\ref{de0410}) takes the form
\begin{equation}
	\label{phiAlphaEq}
	\varphi_{;\mu}= F_{\mu\nu}\xi^\nu.
\end{equation}

It is important to note that the Killing equation (\ref{killingEquation}) does not include the potential $A_\nu$ and therefore it determines a symmetry operator $X$  of the free ($A_\nu=0$) Dirac equation (\ref{dirac1}).

Equation (\ref{phiAlphaEq}) shows that a field $F_{\nu\mu}$ (or a potential $A_\nu$) permitting a symmetry operator $X$ of the form (\ref{symm-oper-zero-1-1}) is determined by the Killing vector field  $\xi^\nu$.

This allows us to use the properties of invariance of the free Dirac equation for classification of symmetry operators $X$ and fields $F_{\nu\mu} $. Below we follow the definitions of equivalence given in \cite{vnsh_ecl1973,bagr_vnsh_ecl1975,bagr_m_sh_vnsh1973}.

The free Dirac equation (\ref{dirac1}) is invariant with respect to the $(2+1)$-dimensional  Poincare group $\mathcal{P}(1,2)$ (translations and the Lorentz group $SO(1,2)$) of transformations of $(2+1)$-dimensional Minkowski spacetime and a transformation of the wave function $\psi$. Denote by $\mathfrak{p}(1,2)$ the Poincare algebra of the Poincare group $\mathcal{P}(1,2)$.

Let $g$ be the generic element of the transformation group $\mathcal{P}(1,2)$ :
\begin{equation}
	\label{group1}
	{x'}^\mu=(g\cdot x)^\mu = a^{\mu\,.}_{.\, \nu}x^\nu+b^\mu,
\end{equation}
where $b^\mu$ are real translation parameters,  and $a^{\mu\,.}_{. \,\nu}$  is a $3\times 3$ real pseudo-orthogonal matrix  defined by $a^{\alpha\,.}_{.\, \nu}g_{\alpha\beta}a^{\beta\,.}_{.\, \mu}=g_{\nu\mu}$. The inverse transformation is
\begin{equation}
	\label{group3}
	x^\nu=(g^{-1} \cdot x{'})^\nu =  a^{.\, \nu}_{\alpha \, .}(x{'}^\alpha-b^\alpha ),
\end{equation}
where we use the notation: $(a^{-1})^{\alpha\,.}_{.\, \nu}=a^{.\,\alpha}_{\nu\, .} $, $ a^{\alpha\,.}_{.\, \nu}a^{.\,\nu}_{\beta \, .}=\delta^\alpha_\beta $.

Let $\psi(x)$ be a solution of the Dirac equation (\ref{dirac1}):
\begin{equation}
	\label{dirac2}
	\hat H(x)\psi(x) = 0,\quad
	\hat H(x) = \gamma^\nu
	\hat{\mathcal{P}}_\nu - m.
\end{equation}
Let us transform equation (\ref{dirac2}) by using the  change of variables (\ref{group3}) and the function
\begin{equation}
	\label{dirac3}
	\psi(x)=S(g)\bar\psi(x'),
\end{equation}
where the matrix $S(g)$ is determined by
\begin{equation}
	\label{dirac4}
	S^{-1}(g)a^{\nu\, .}_{.\,\alpha}\gamma^\alpha S(g)=\gamma^\nu .
\end{equation}
As a result,  equation (\ref{dirac2}) takes the form
\begin{gather}
	\label{dirac5}
	\hat{\bar H}(x')\bar\phi(x') = 0,\quad
	\hat{\bar H}(x') = \gamma^\nu \hat{\mathcal{P}}'_\nu - m,\\ \nonumber
  	     \hat{\mathcal{P}}'_\nu = p'_\nu - A_{(g)\nu} (x'),\quad 	
	     p'_\nu=i \frac{\partial}{\partial {x'}^\nu},\\
	\label{dirac6}
	A_{(g)\nu} (x')= a^{.\, \alpha}_{\nu \, .} A_\alpha (x)= a^{.\, \alpha}_{\nu \, .} A_\alpha (g^{-1}x').
\end{gather}

\begin{definition}
The potentials $A_\nu(x)$  and $A_{(g)\nu} (x)$ bound  by condition (\ref{dirac6}) are called equivalent with respect to the group $\mathcal{P}(1,2)$.
\end{definition}

From (\ref{dirac2}), (\ref{dirac5}), and  (\ref{dirac6}) it follows that if a function $\psi(x)$ is a solution of the Dirac equation
(\ref{dirac2}) with  potential $A_\nu(x)$, then $\bar\psi (x)$ will be a solution of the Dirac equation with  potential
$A_{(g)\nu} (x)$.

For this reason, we can  consider only non-equivalent potentials.

After the change of variables (\ref{group1}) and  transformation of the wave function (\ref{dirac3}), the symmetry operator  $X$  in  (\ref{symm-oper-zero-1-1}) takes the form
\begin{gather}
	\label{equiv-symm-1}
	\bar X = \bar X(x', \mathcal{P}') = S^{-1}(g)X(x,\mathcal{P})S(g) = \\ \nonumber
	= a^{\alpha \, .}_{.\, \nu } a^{\beta \, .}_{.\, \mu }A^{\nu\mu}\big( L'_{\alpha\beta}-
b_\alpha\mathcal{P}'_\beta + b_\beta\mathcal{P}'_\alpha\big)+
a^{\alpha \, .}_{.\, \nu }B^{\nu} \,\mathcal{P}'_\alpha +\varphi(g^{-1} x')+\varphi_\alpha(g^{-1}x')\gamma^\alpha,
\end{gather}
where $L'_{\alpha\beta}= x'_\alpha\mathcal{P}'_\beta - x'_\beta \mathcal{P}'_\alpha$.

\begin{definition}
	We call the symmetry operators $X(x,\mathcal{P})$ and $\bar X(x,\mathcal{P})$ in (\ref{equiv-symm-1}) equivalent  with respect to the group $\mathcal{P}(1,2)$.
\end{definition}

Then it is sufficient to consider only nonequivalent operators $X$.

Let us show that classification of the symmetry operators $X$ can be reduced to classification of one-dimensional subalgebras of the Poincare algebra   $\mathfrak{p}(1,2)$ under the adjoint action of the Poincare group $\mathcal{P}(1,2)$.

From (\ref{symm-oper-zero-1-1}) it follows that there must be a one-to-one correspondence between the symmetry operator $X$ and an element $\eta$ of the  Poincare algebra $\mathfrak{p}(1,2)$:
\begin{equation}
	\label{one-to-one-1}
	X = A^{\alpha\nu}L_{\alpha\nu}+B^\nu \mathcal{P}_\nu + \varphi + \varphi_\alpha\gamma^\alpha \Leftrightarrow
	\eta = A^{\alpha\nu}l_{\alpha\nu}+B^\nu p_\nu,
\end{equation}
where $l_{\mu\nu} =x_\mu p_\nu- x_\nu p_\mu$. From (\ref{group1})--- (\ref{group3}), (\ref{dirac3}), (\ref{dirac4}), and (\ref{equiv-symm-1}) it is seen that  one-to-one correspondence (\ref{equiv-symm-1}) is also valid for $\bar X$ in (\ref{equiv-symm-1}) and $\bar \eta$,
\begin{eqnarray}
	\label{one-to-one-2}
	\bar X \Leftrightarrow \bar \eta && =  \bar \eta (x', p')=S^{-1}(g)\eta(x,p) S(g) =\\ \nonumber
	&& = a^{\alpha \, .}_{.\, \nu } a^{\beta \, .}_{.\, \mu }A^{\nu\mu}\big( l'_{\alpha\beta}-
b_\alpha p'_\beta + b_\beta p'_\alpha\big)+a^{\alpha \, .}_{.\, \nu }B^{\nu} \, p'_\alpha.
\end{eqnarray}
This  correspondence indicates that classification of the symmetry operators $X$ is equivalent to classification of the elements $\eta$ of the Poincare algebra  $\mathfrak{p}(1,2)$ under transformations of the form (\ref{one-to-one-2}) or, in other words, classification of  orbits in the Lie algebra $\mathfrak{p}(1,2)$ of the Lie group $\mathcal{P}(1,2)$.

Let us first assume that $A^{\alpha\nu}\ne 0$ in $\eta$ in (\ref{one-to-one-1}) and classify $\xi=A^{\nu\mu}l_{\nu\mu}$ under transformation (\ref{one-to-one-2}). This implies seeking for the orbits in the Lie algebra $\mathfrak{so}(1,2)$.

Nonequivalent representatives of orbits are well known (e.g., \cite{vnsh_ecl1973,bagr_vnsh_ecl1975,bagr_m_sh_vnsh1973} and \cite{mukunda-1974,gilmore-1974})   and can be written up to a constant factor as
\begin{eqnarray}
\label{equiv-symm-5}
&&  l_{01},\quad l_{21}, \quad l_{21}+l_{01}.
\end{eqnarray}
For each element of $\mathfrak{so}(1,2)$ in (\ref{equiv-symm-5}), we have
an element  of the  Poincare algebra $\mathfrak{p}(1,2)$:
$l_{01}+B^\alpha p_\alpha$,  $l_{21}+B^\alpha p_\alpha$, $l_{21}+l_{01}+B^\alpha p_\alpha$ where $B^\alpha$ are arbitrary
constants. Then,  using translations $x'^\alpha=x^\alpha+b^\alpha $, we can simplify these elements to
$l_{01}+a p_2$,  $l_{21}+a p_0$, $l_{21}+l_{01}+a (p_0-p_2)$, where $a$ is an arbitrary constant.

 If $A^{\alpha\nu}= 0$  in (\ref{one-to-one-1}), then $\eta=B^\alpha p_\alpha$. Nonequivalent  $\eta$ under adjoint action
 of $SO(1,2)$ are obtained directly: $p_0$, $p_1$, $p_0+p_1$.

 Summarising, we can write the representatives of orbits in the Poincare algebra $\mathfrak{p}(1,2)$:
 \begin{eqnarray}
\label{equiv-symm-6}
&& p_0, \quad p_1, \quad p_0+p_1, \cr
&&l_{21}+a p_0, \quad l_{01}+a p_2,\quad  \quad  l_{01}+l_{21}+a (p_0-p_2).
\end{eqnarray}

 For separation of variables in the Dirac equation (\ref{dirac1}), we need sets of pairs of mutually commuting symmetry operators $\{X_1, X_2\}$ of the form (\ref{symm-oper-zero-1-1}),  $[X_1, X_2]=0$.

 \begin{definition}
We call the complete set of symmetry operators for the Dirac equation (\ref{dirac1}) a pair $\{X_1, X_2\}$ of mutually commuting linearly independent symmetry operators $X_1$ and $X_2$ nonequivalent with respect to the group $\mathcal{P}(1,2)$. 	
 \end{definition}

\begin{definition}
	  Consider two sets of symmetry operators $\{X_1, X_2\}$ and $\{\tilde X_1, \tilde X_2\}$. We call these sets equivalent if
	 \begin{eqnarray*}
		&& \tilde X_j=c^k_j \bar X_k(x,\mathcal{P})+c_j , \quad j,k=1,2,
	\end{eqnarray*}
 where $\bar X_k(x,\mathcal{P})$ are  defined by equation   (\ref{equiv-symm-1}), and $c^k_j, c_j$ are constants, $\det (c^k_j)\ne 0$.
\end{definition}

 Taking into account the one-to-one correspondence (\ref{one-to-one-1}), we can
 classify the  pairs $\{\eta_1, \eta_2\}$ of elements of the Poincare algebra $\mathfrak{p}(1,2)$ commuting as
 vector fields on Minkowski spacetime.

 \begin{definition}
	  Two sets  $\{\tilde\eta_1,\tilde\eta_2\}$ and $\{\eta_1, \eta_2\}$ are regarded as equivalent if
	  \begin{eqnarray*}
		&& \tilde\eta_j(x,p)=c^k_j \bar\eta_k(x,p)+c_j ,
	\end{eqnarray*}
	 where  $\bar\eta_k(x,p)$ are defined by equation (\ref{one-to-one-2}).
 \end{definition}

From (\ref{equiv-symm-6}) we  directly obtain nonequivalent sets of commuting vector fields  $\{\eta_1, \eta_2\}$, which are presented in the  table \cite{bagr_m_sh_vnsh1973}
   \begin{equation}
       \begin{array}{c|c|c}
     &  \eta_1 &\eta_2\\[2mm]\hline
    1.&  p_0 & p_1 \\
    2.&   p_1 & p_2 \\
    3.&   p_0 & l_{21} \\
    4.&   p_1 & l_{02} \\
    5.&   p_1 & \frac 12(p_0 + p_2) \\
    6.&   \frac 12(p_0 + p_1) & l_{21}+l_{01}+a(p_0-p_2) \\
     \end{array}
     \label{complete-sets-1}
   \end{equation}
For each pair $\{\eta_1, \eta_2\}$  in (\ref{complete-sets-1}), according to  one-to-one
correspondence (\ref{one-to-one-1}),  we immediately obtain a pair $\{X_1, X_2\}$.

In the next section, we use the nonequivalent sets of symmetry operators $\{X_1, X_2\}$ for separation of variables in the Dirac equation (\ref{dirac1}).

\section{Separation of variables }
\label{separation}

We start with a simple lemma.
\begin{lemma}
\label{lemma}
Let the Dirac equation ({\ref{dirac1}}) permits a symmetry operator $X$ of the form (\ref{symm-oper-zero-1-1}), i.e. conditions (\ref{de0410}) and (\ref{killingEquation}) hold. Then, in a coordinate system where $\xi^\mu=\delta^\mu_0$, we have $\displaystyle \partial_0 F_{\alpha\beta}=0$.
\end{lemma}
\begin{proof}
	The proof of this lemma is quite straightforward. Consider equation (\ref{phiAlphaEq}). For $\xi^\mu=\delta^\mu_0$ we have $\varphi_{;\mu}=F_{\mu\nu}\xi^\nu= F_{\mu0}$. Then, from the compatibility condition, $ \varphi_{;0 1}=\varphi_{;1 0}$, it follows that $\displaystyle\partial_0 F_{0 1}=0$, and from  $ \varphi_{;0 2}=\varphi_{;2 0}$ we obtain $\displaystyle \partial_0 F_{0 2}=0$. The relation $\varphi_{; 1 2}=\varphi_{; 2 1}$ yields $\displaystyle\partial_2 F_{10}=\displaystyle \partial_1 F_{20}$. Taking into account the Maxwell equation $\displaystyle \partial_2 F_{0 1}+\partial_0 F_{12} +\partial_1 F_{2 0}=0$, we arrive at the following relation: $\displaystyle \partial_0 F_{12}=0$.
\end{proof}

From this lemma we get the following corollary. Since $F_{\mu\nu}= \displaystyle \partial_\nu A_{\mu}-\partial_\mu A_{\nu}$ and $A_\mu=A_\mu(x^1,x^2)$,  we find that $\varphi_{;\mu}= \displaystyle \partial_\mu A_{0}$ and then we get $\varphi=A_0$.

Therefore, the symmetry operator (\ref{opXsymm}) takes the form
\begin{gather}\nonumber
	X=\xi^{\nu}\mathcal{P}_\nu+  \varphi+\varphi_\alpha \gamma^\alpha =
	\xi^{\nu}(p_\nu-A_\nu) +  \varphi+\varphi_\alpha \gamma^\alpha = \\ \nonumber
	 = \delta^{\nu}_0(p_\nu-A_\nu) +  A_0+\varphi_\alpha \gamma^\alpha =
	p_0+\varphi_\alpha \gamma^\alpha.
\end{gather}
From (\ref{de032}) and $\xi^\nu=\delta^\nu_0$ we have that $\varphi_\alpha =0$ and $X=p_0$.

Our next step is finding the  potentials permitting  the sets of symmetry operators  $\{ X_1,X_2 \}$ corresponding to
pairs $\{\eta_1,\eta_2 \}$ of the form (\ref{complete-sets-1}) and  separation of variables in the Dirac equation (\ref{dirac1}).

The results can be formulated in the form of the following theorem.
\begin{theorem}
Each complete set $\{\hat X_1, \hat X_2\}$ of the symmetry operators leads to complete separation of variables in the  Dirac equation in $(2+1)$-dimensional Minkowski spacetime.
\end{theorem}

The proof of this theorem consists in the actual separation of variables with  complete sets of symmetry operators.

In the following subsections we carry out  complete separation of variables in the $(2+1)$-dimensional Dirac equation (\ref{dirac2}) with each complete set of symmetry operators listed in  to table \ref{complete-sets-1}. Also, we  find  electromagnetic potentials permitting separation of variables.

Below we  use curvilinear coordinates denoted by $(u)=(u_0,u_1,u_2)$ with lower indices used for convenience; the partial derivatives are denoted by $\partial / \partial u_\mu = \partial_{u_\mu}$.

Separable solutions of the Dirac equation in separable coordinates (curvilinear in general case) are found from the system
 \begin{equation}
 X_1 \psi =\lambda_1 \psi, \quad X_2 \psi =\lambda_2 \psi, \quad  H\psi=0.
 \label{set1-4-2}
 \end{equation}
Note that if we select the local frame as
\begin{equation}
	e^\mu_a(x) = \partial_a u_\mu,\quad
	e^a_\mu(u) = \partial_{u_\mu}x^a(u) ,
	\label{ddframe}
\end{equation}
then the spinor connection (\ref{spinconn3}) is zero $(\Gamma_\nu=0)$ and the Dirac equation in curvilinear coordinates
$(u)$ is
 \begin{eqnarray}
 &&H\psi= \left\{\gamma^\nu \Big(i\partial_{u_\nu}-A_\nu(u)\Big) - m\right\}\psi=0.
\label{set1-4-2-A}
\end{eqnarray}
The solution $\psi$ of system (\ref{set1-4-2}) and the solution $\phi_C$ of the Dirac equation in Cartesian coordinates $(x)$,
 \begin{equation}
H_C\phi_C=\left\{\hat\gamma^\alpha \left(i \partial_\alpha-A_{(C)\alpha}(x)\right)-m\right\}\phi_C=0,
 \label{set1-4-12}
 \end{equation}
are bound by means of an operator $\hat S$
 \begin{equation}
\phi_C=\hat S \psi,
 \label{s-operator}
 \end{equation}
 where $A_{(C)\alpha}$ are the Cartesian components of the potential, and $\hat\gamma^\nu$ are given by (\ref{gammamatr2}).

\subsection{The set $\{ p_0, p_1 \} $ }
\label{subsect1}
For the set $\{ p_0, p_1 \} $  (see  table \ref{complete-sets-1}), in view of one-to-one correspondence  (\ref{one-to-one-1}), symmetry operators $\{ X_1,X_2 \}$ of the form
(\ref{symm-oper-zero-1-1})  can be written as
\begin{eqnarray}
\label{set1}
X_1=  \mathcal{P}_0 + \underset{1}{\varphi} + {\underset{1}{\varphi}}\,_\alpha \gamma^\alpha , \quad
X_2= \mathcal{P}_1 + \underset{2}{\varphi} + {\underset{2}{\varphi}}\,_\alpha \gamma^\alpha
\end{eqnarray}
Here, according to  (\ref{killing-1}), (\ref{phiAlphaEq}) and (\ref{de032}), we have, in Cartesian coordinates $(x^0,x^1,x^2)$, that ${\underset{1}{\xi}}\,^\nu=\delta^\nu_0$, $ {\underset{2}{\xi}}\,^\nu=\delta^\nu_1$, and
\begin{eqnarray}
&&  {\underset{1}{\varphi}}\,_\nu = -\displaystyle\frac{1}{4} s e_{\nu\mu\sigma}\delta^\mu_{0;\alpha}g^{\sigma\alpha}=0,
 \quad
 {\underset{2}{\varphi}}\,_\nu = -\displaystyle\frac{1}{4} s e_{\nu\mu\sigma}\delta^\mu_{1;\alpha}g^{\sigma\alpha}=0, \nonumber \\
&&  {\underset{1}{\varphi}}\,_{;\mu} = F_{\mu 0}, \quad  {\underset{2}{\varphi}}\,_{;\mu} = F_{\mu 1} . \label{set1-2}
 \end{eqnarray}
From (\ref{set1-2}),  in view of  Lemma \ref{lemma}, we obtain   $F_{\mu\nu}=F_{\mu\nu}(x^2)$ and
\begin{equation*}
A_\nu=A_\nu (x^2), \quad \underset{1}{\varphi} = A_0, \quad  \underset{2}{\varphi} = A_1,
 \end{equation*}
where $A_\nu (x^2)$ are arbitrary functions of $x^2$.

Then from (\ref{set1}) it follows that
\begin{equation*}
X_1= i\partial_0, \quad X_2=i \partial_1.
 \end{equation*}
Separable solutions to the Dirac equation (\ref{dirac1}) are obtained from system (\ref{set1-4-2}),
where $H=H_C $, and the eigenvalues $\lambda_1$ and $\lambda_2$ play the role of arbitrary constants of separation of variables.
 Integrating equations (\ref{set1-4-2}), we can write $\psi$ as
 \begin{equation}
\label{set1-5}
\psi= \phi_C = e^{-i\lambda_1 x^0-i\lambda_2 x^1}\tilde\psi(x^2),\quad
\tilde\psi(x^2)=
\left(
\begin{array}{c}
     \tilde\psi_1(x^2)\\
    \tilde\psi_2(x^2)\\
     \end{array}
\right).
 \end{equation}
Note that, in the case under consideration,  the Cartresian coordinates are separable  and $\psi=\phi_C$.
Substituting (\ref{set1-5}) in the Dirac equation  (\ref{set1-4-12}), we have
 \begin{gather}
\label{set1-7}
\bigg{\{}\hat\gamma^0\big(\lambda_1-A_0(x^2) \big)+ \hat\gamma^1\big(\lambda_2-A_1(x^2) \big)+\\ \nonumber
+\hat\gamma^2 \big(i\displaystyle\frac{d}{d x^2} -A_2(x^2) \big)-
m \bigg{\}}\tilde\psi(x^2)=0,
 \end{gather}
where (\ref{set1-7}) is a system of ordinary differential equations (ODEs).

\subsection{The set $\{  p_1,  p_2  \} $ }
\label{subsect2}

Separation of variables with the set $\{ p_1, p_2 \}$ (see table \ref{complete-sets-1}) is quite similar to that considered in   subsection \ref{subsect1}.

Therefore, we omit intermediate details and present the main results.

The symmetry operators $\{ X_1,X_2 \}$ corresponding to $\{ p_1, p_2 \}$, in Cartesian coordinates $(x^0,x^1,x^2)$, are
\begin{eqnarray}
\label{set2}
X_1=  \mathcal{P}_1 + \underset{1}{\varphi} + {\underset{1}{\varphi}}\,_\alpha \gamma^\alpha , \quad
X_2= \mathcal{P}_2 + \underset{2}{\varphi} + {\underset{2}{\varphi}}\,_\alpha \gamma^\alpha .
\end{eqnarray}
Here, in accordance with (\ref{killing-1}), we find   ${\underset{1}{\xi}}\,^\nu=\delta^\nu_1$,
${\underset{2}{\xi}}\,^\nu=\delta^\nu_2$. From (\ref{de032}), (\ref{phiAlphaEq}) it follows that
${\underset{1}{\varphi}}\,_\nu =$ ${\underset{2}{\varphi}}\,_\nu=0$ and ${\underset{1}{\varphi}}\,_{;\mu} = F_{\mu 1}$,  ${\underset{2}{\varphi}}\,_{;\mu} = F_{\mu 2}$. According to Lemma \ref{lemma}, we obtain $F_{\mu\nu}=F_{\mu\nu}(x^0)$, and a potential admitting  the set of symmetry operators (\ref{set2}) is
\begin{equation*}
A_\nu=A_\nu (x^0), \quad \underset{1}{\varphi} = A_1, \quad  \underset{2}{\varphi} = A_1,
 \end{equation*}
where $A_\nu=A_\nu (x^0)$ are arbitrary functions of $x^0$. From (\ref{set2}) it follows that
\begin{equation}
\label{set2-1-4}
	X_1 = i\partial_1, \quad X_2 = i\partial_2.
 \end{equation}
Separable solutions $\psi$ to the Dirac equation  obtained from  system (\ref{set1-4-2})
(where $H=H_C$) with  operators (\ref{set2-1-4}) are
 \begin{gather}
\psi=\phi_C=e^{-i\lambda_1 x^1-i\lambda_2 x^2}\tilde\psi(x^0),\quad
\tilde\psi(x^0)=
\left(
\begin{array}{c}
     \tilde\psi_1(x^0)\\
     \tilde\psi_2(x^0)\\
     \end{array}
\right), \nonumber \\ \nonumber
\label{set2-5}
\left\{\hat\gamma^0 \big(i\displaystyle\frac{d}{d x^0}-A_0(x^0)\big) + \hat\gamma^1 \big(\lambda_1-A_1(x^0) \big)+ \hat\gamma^2\big(\lambda_2-A_2(x^0) \big)
 - m \right\}\tilde\psi(x^0)=0.
 \end{gather}

 \subsection{The set $\{  p_0,  l_{21} \} $ }
\label{subsect3}
Consider the  set $\{ p_0,  l_{21} \}$ (see  table \ref{complete-sets-1}). Similar to the above cases,
in accordance with (\ref{one-to-one-1}),  we come to   symmetry operators $\{ X_1,X_2 \}$ of the form
\begin{eqnarray}
\label{set3}
X_1=  \mathcal{P}_0 + \underset{1}{\varphi} + {\underset{1}{\varphi}}\,_\alpha \gamma^\alpha , \quad
X_2= L_{21} + \underset{2}{\varphi} + {\underset{2}{\varphi}}\,_\alpha \gamma^\alpha .
\end{eqnarray}
Then we  can write  the Killing vector fields  (\ref{killing-1}) in Cartesian coordinates $(x^0,x^1,x^2)$ as $ {\underset{1}{\xi}}\,^\nu=\delta^\nu_0$, ${\underset{2}{\xi}}\,^\nu=x^1\delta^\nu_2-x^2\delta^\nu_1$, and ${\underset{j}{\varphi}}\,_\nu$,  $\underset{j}{\varphi}$ are found from equations (\ref{de032}) and (\ref{phiAlphaEq}), respectively.

To straighten out the  Killing vector field  ${\underset{2}{\xi}}^\nu$, we introduce polar coordinates $(x^0, r,\varphi )$, where
\begin{eqnarray}
\label{set3-1}
x^1=r \cos\varphi, \quad x^2=r\sin\varphi , \,\,\, \quad r=\sqrt{(x^1)^2+(x^2)^2},\quad \varphi =\arctan{\Big(\frac{x^2}{x^1}\Big)} .
\end{eqnarray}
In   coordinates (\ref{set3-1}), the components of the metric tensor are $(g_{\mu\nu})=\mbox{diag}(1,-1,-r^2)$, the Killing vectors $\underset{j}{\xi^\nu}$ take the Kronecker delta form $\underset{1}{\xi}^\nu=\delta^\nu_0$, $\underset{2}{\xi}^\nu=\delta^\nu_2$, and the non-zero  Christoffel symbols read $\Gamma^1_{22}=\Gamma^r_{\varphi\varphi}=-r$, $\Gamma^2_{12}=\Gamma^\varphi_{r\varphi}=r^{-1}$.
Let us define   local frame fields (\ref{triada1}) as
\begin{equation*}
 (e^\nu_a (x))= \big(\delta_0^\nu, \delta_1^\nu , r^{-1}\delta_2^\nu  \big),  \quad
 (e_\nu^a (x))= \big(\delta^0_\nu, \delta^1_\nu , r\delta^2_\nu  \big).
\end{equation*}
Then the spinor connection (\ref{spinconn3}) in polar coordinates takes the form
\begin{equation}
\label{set3-4}
\Gamma_0=\Gamma_1=0, \quad \Gamma_2=\displaystyle\frac{1}{4}[\hat\gamma^1,\hat\gamma^2]=
-\displaystyle\frac{i}{2} s\hat\gamma^0 = -\displaystyle\frac{i}{2} s\sigma_3.
\end{equation}
From (\ref{de032}) and (\ref{phiAlphaEq}) we have ${\underset{1}{\varphi}}\,_\nu=0$, ${\underset{2}{\varphi}}\,_\nu=-(s/2)\delta_{0\nu}$, where we have used that $e_{\mu\nu\alpha}= \sqrt{g}\epsilon_{\mu\nu\alpha}=r\epsilon_{\mu\nu\alpha}$ according to (\ref{l-c-tensor1}).
Taking into account Lemma \ref{lemma}, we find
\begin{equation}
\label{set3-8}
A_\nu=A_\nu(r), \quad {\underset{1}{\varphi}}=A_0(r), \quad \underset{2}{\varphi}=A_2(r).
\end{equation}
Here $A_\nu(r)$ are arbitrary functions of $r$.
In view of (\ref{set3-1})---(\ref{set3-8}), we can write  the symmetry operators (\ref{set3}) and the Dirac operator (\ref{dirac1}) in  polar coordinates (\ref{set3-1}) as
\begin{gather}
X_1=i\displaystyle\partial_0, \quad  X_2=i\displaystyle\partial_\phi, \label{set3-9-1}\\
H=\hat\gamma^0\Big(i\displaystyle\partial_0 - A_0(r)\Big)+\hat\gamma^1\Big( i\displaystyle\partial_r-A_1(r)\Big)+ \displaystyle\frac{1}{r}\hat\gamma^2\Big(i\displaystyle\partial_\varphi + \displaystyle\frac{s}{2} \hat\gamma^0-A_2(r) \Big) - m.
\label{set3-9}
\end{gather}

Separable solutions $\psi$ to the Dirac equation in polar coordinates  are obtained from the system (\ref{set1-4-2}) with operators (\ref{set3-9-1}) and (\ref{set3-9}) in  the form
\begin{gather}
	\label{set3-10}
	\psi=e^{-i\lambda_1 x^0-i\lambda_2 \varphi}\tilde\psi(r),\quad
	\tilde\psi(r)=
		\left(
			\begin{array}{c}
			     \tilde\psi_1(r)\\
			     \tilde\psi_2(r)\\
			\end{array}
		\right),\nonumber \\ \nonumber
\label{set3-12}
\Big\{\hat\gamma^0\big( \lambda_1 -A_0(r)\big)+\hat\gamma^1 \big(i\displaystyle\frac{d}{d r} - A_1(r)\big) +\displaystyle\frac{1}{r}\hat\gamma^2\big(\lambda_2+ \displaystyle\frac{s}{2} \hat\gamma^0-A_2(r) \big)- m\Big\}\psi(r)=0.
 \end{gather}

\textbf{Catresian coordinates}.
In Cartesian coordinates $(x^0,x^1,x^2)$, the solution $\phi_C$ satisfies the Dirac equation
(\ref{set1-4-12}) and is connected with $\psi$ from (\ref{set3-10}), according to (\ref{s-operator}), as follows:
\begin{equation*}
\phi_C=\hat S\psi=\exp\big(\displaystyle\frac{s\varphi}{2i} \hat\gamma^0\big) \psi.
\end{equation*}
The Cartesian components  $A_{(C)\alpha}$ of the potential  are
\begin{gather*}
	A_{(C) 0} =A_0(r), \quad
	A_{(C) 1} =\displaystyle\frac{x^1}{r} A_1(r)-\displaystyle\frac{x^2}{r^2} A_2(r),\quad
	A_{(C) 2} =\displaystyle\frac{x^2}{r} A_1(r)+\displaystyle\frac{x^1}{r^2} A_2(r).
\end{gather*}

\subsection{The set $\{ p_1 , l_{02} \} $ }
\label{subsect4}

Separation of variables with the  set $\{ p_1, l_{02} \}$ (see table \ref{complete-sets-1})
is similar to the set $\{ p_0, l_{21} \} $ considered in previous subsection \ref{subsect3}.

In accordance with (\ref{one-to-one-1}),  we get the  symmetry operators $\{ X_1,X_2 \}$ in Cartesian coordinates
$(x^0,x^1,x^2)$ as
\begin{eqnarray}
\label{set4}
X_1=  \mathcal{P}_1 + \underset{1}{\varphi} + {\underset{1}{\varphi}}\,_\alpha \gamma^\alpha , \quad
X_2= L_{02} + \underset{2}{\varphi} + {\underset{2}{\varphi}}\,_\alpha \gamma^\alpha .
\end{eqnarray}
The Killing vector fields (\ref{killing-1})  take the form ${\underset{1}{\xi}}\,^\nu=\delta^\nu_1$, ${\underset{2}{\xi}}\,^\nu=x^0\delta^\nu_2+x^2\delta^\nu_0$, and ${\underset{j}{\varphi}}\,_\nu$,  $\underset{j}{\varphi}$, $j=1,2$, are found from equations (\ref{de032}) and (\ref{phiAlphaEq}), respectively.

We carry out the separation of variables with  operators (\ref{set4}) in two domains: $|x^0|>|x^2|$ and  $|x^0|<|x^2|$.

\subsubsection*{The domain $|x^0|>|x^2|$.}

Let us introduce a coordinate system $(u_0,u_1,u_2)$ as
\begin{eqnarray}
\label{set4-1}
&& u_0=\sqrt{(x^0)^2 - (x^2)^2  } ,  \quad
u_2=\displaystyle\frac{1}{2}\log \left( \displaystyle\frac{x^0+x^2}{x^0-x^2}\right), \quad u_1=x^1.
\end{eqnarray}
 The inverse coordinate transformation is given by
  \begin{eqnarray}
\label{set4-1-A}
&& x^0=\varepsilon u_0 \cosh u_2, \quad x^2=\varepsilon u_0 \sinh u_2,\quad
\varepsilon=\operatorname{sign}(x^0+x^2),\quad u_0>0.
\end{eqnarray}
The Killing vectors in coordinates (\ref{set4-1}) take the Kronecker delta form $\underset{1}{\xi}^\nu=\delta^\nu_1, \quad \underset{2}{\xi}^\nu=\delta^\nu_2$.
In coordinates (\ref{set4-1}), the  components of the metric tensor $g_{\mu\nu}$ are $(g_{\mu\nu})=\mbox{diag}(1,-1,-u_0^2)$, and the non-zero Christoffel symbols take the form
\begin{equation}
	\label{crr41}
	\Gamma^0_{22}=u_0, \quad \Gamma^2_{02}=u_0^{-1}.
\end{equation}
We specify the local frame fields (\ref{triada1}) according to (\ref{ddframe}):
\begin{eqnarray*}
	&& e^\nu_0=\left(\displaystyle\partial_0 u_\nu\right) = \varepsilon\big(\cosh u_2\,\delta_0^\nu -u_0^{-1}\sinh u_2 \, \delta_2^\nu\big), \quad e^\nu_1=\left(\displaystyle\partial_1 u_\nu\right)=\delta_1^\nu, \cr
	&& e^\nu_2 =\left(\displaystyle\partial_2 u_\nu\right)=\varepsilon\big(-\sinh u_2 \,\delta_0^\nu + u_0^{-1}\cosh u_2 \, \delta_2^\nu\big).
\end{eqnarray*}
From (\ref{de032}) we have ${\underset{1}{\varphi}}\,_\nu=0$, ${\underset{2}{\varphi}}\,_\nu=-(s/2) \delta_{1\nu}$. Following  Lemma \ref{lemma}, we find the potential permitting the set of symmetry operators (\ref{set4}) as
\begin{equation}
	A_\nu=A_\nu(u_0)
	\label{pot41}
\end{equation}
and $\underset{1}\varphi=A_1$, $\underset{2}\varphi=A_2$. Here $A_\nu(u_0)$ are arbitrary functions of $u_0$.

Then we can write the symmetry operators
 (\ref{set4})
 in the form
\begin{eqnarray}
&& X_1=i\displaystyle\partial_{u_1}, \quad  X_2=i\displaystyle\partial_{u_2} - \displaystyle\frac{s}{2}\gamma^1. \label{set4-1-9}
\end{eqnarray}
The Dirac equation (\ref{dirac1}) in coordinates (\ref{set4-1}), (\ref{set4-1-A}) takes the form (\ref{set1-4-2-A}), where it is necessary to put  $\gamma^1=\hat\gamma^1$.

 Separable solutions to the Dirac equation (\ref{dirac1}) in  coordinates  (\ref{set4-1})
 are found from system (\ref{set1-4-2}), (\ref{set1-4-2-A})
with operators (\ref{set4-1-9})  and  can be written as
  \begin{gather}
	\label{set4-10}
	\psi=\frac{1}{\sqrt{u_0}}\exp\left(-i\lambda_1 u_1-i\lambda_2 u_2-i\frac{s}{2}\hat\gamma^1 u_2 \right)\tilde\psi(u_0),\quad
	\tilde\psi(u_0)=
		\left(
			\begin{array}{c}
			     \tilde\psi_1(u_0)\\
			     \tilde\psi_2(u_0)\\
			\end{array}
		\right),\\ \nonumber
	\Big\{ \varepsilon\hat\gamma^0 \Big(i\displaystyle\frac{d}{d u_0}-A_0(u_0)\Big) + \hat\gamma^1 \Big(\lambda_1- A_1(u_0)\Big)+ \\ \nonumber
	+\varepsilon u_0^{-1}\hat\gamma^2 \Big( \lambda_2 -A_2 (u_0)\Big) - m \Big\}\psi(u_0)=0.
 \end{gather}

\textbf{Catresian coordinates}.
The solution $\phi_C$ of  the Dirac equation (\ref{set1-4-12}) in Cartesian coordinates $(x^0,x^1,x^2)$  and the solution $\psi$ from (\ref{set1-4-2}) in the curvilinear separable coordinates (\ref{set4-1}), (\ref{set4-1-A}) are  $\phi_C=\psi$.

The Cartesian components  $A_{(C)\alpha}$ of the potential are
\begin{gather*}
	A_{(C) 0} =\displaystyle\frac{x^0}{u_0}A_0(u_0)-\displaystyle\frac{x^2}{u_0^2}A_2(u_0), \quad
	A_{(C) 1} =A_1(u_0),\\
	A_{(C) 2} =-\displaystyle\frac{x^2}{u_0}A_0(u_0) + \displaystyle\frac{x^0}{u_0^2}A_2(u_0).
\end{gather*}

\subsubsection*{The domain $|x^0|<|x^2|$.}
The coordinate system $(u_0,u_1,u_2)$ for this domain is as follows:
\begin{eqnarray}
\label{set4-2-1}
&& u_0=\sqrt{(x^2)^2 - (x^0)^2  }, \quad
u_2=\displaystyle\frac{1}{2}\log \left( \displaystyle\frac{x^2+x^0}{x^2-x^0}\right), \quad u_1=x^1.
\end{eqnarray}
 The inverse coordinate transformation is
  \begin{eqnarray}
\label{set4-2-A}
&& x^0=\varepsilon u_0 \sinh u_2, \quad x^2=\varepsilon u_0 \cosh u_2,\quad
\varepsilon=sign(x^2+x^0),\quad
u_0>0.
\end{eqnarray}
The Killing vectors in coordinates (\ref{set4-2-1}), (\ref{set4-2-A}) take the Kronecker delta form $\underset{1}{\xi}^\nu=\delta^\nu_1, \quad \underset{2}{\xi}^\nu=\delta^\nu_2$.
The the metric tensor $(g_{\mu\nu})$ is $(g_{\mu\nu})=\mbox{diag}(-1,-1, u_0^2)$. The non-zero  Christoffel symbols are the same as in the previous case and are given by (\ref{crr41}).
Choose  local frame fields (\ref{triada1}) in the form
\begin{eqnarray}
\label{set4-2-3}
 && e^\nu_0 =\left(\displaystyle\partial_0 u_\nu\right)=\varepsilon\big(-\sinh u_2\,\delta_0^\nu + u_0^{-1}\cosh u_2 \, \delta_2^\nu\big), \quad
    e^\nu_1=\left(\displaystyle\partial_1 u_\nu\right)=\delta_1^\nu, \cr
 && e^\nu_2 =\left(\displaystyle\partial_2 u_\nu\right)=\varepsilon\big(\cosh u_2 \, \delta_0^\nu - u_0^{-1}\sinh u_2 \, \delta_2^\nu\big).
\end{eqnarray}
Relation (\ref{pot41}) also remains valid. From (\ref{set4-2-1})---(\ref{set4-2-3}) 
we obtain that symmetry operators (\ref{set4}) and the Dirac operator in coordinates (\ref{set4-2-1}), (\ref{set4-2-A}) take the form (\ref{set4-1-9}) and (\ref{set1-4-2-A}), respectively.

Separable solutions to the Dirac equation (\ref{dirac1}) in  coordinates  (\ref{set4-2-1}), (\ref{set4-2-A}) are  given by (\ref{set4-10}), where $\tilde\psi(u_0)$ satisfies the equation
\begin{gather}
	\label{set4-2-112}
	\Big\{ \varepsilon\hat\gamma^2 \Big(i\displaystyle\frac{d}{d u_0}-A_0(u_0)\Big) + \hat\gamma^1 \Big(\lambda_1- A_1(u_0)\Big)+ \\ \nonumber
	+\varepsilon u_0^{-1}\hat\gamma^0 \Big( \lambda_2 -A_2 (u_0)\Big) - m \Big\}\psi(u_0)=0.
\end{gather}

\textbf{Catresian coordinates}.
The solution $\phi_C$ of  the Dirac equation (\ref{set1-4-12}) in Cartesian coordinates $(x^0,x^1,x^2)$  and the solution $\psi$ from (\ref{set1-4-2}) in the curvilinear separable coordinates (\ref{set4-2-1}), (\ref{set4-2-A}) are identical: $\phi_C=\psi$. The Cartesian components  $A_{(C)\alpha}$ of the potential  are
\begin{gather*}
	A_{(C) 0} =-\displaystyle\frac{x^0}{u_0}A_0(u_0) + \displaystyle\frac{x^2}{u_0^2}A_2(u_0), \quad
	A_{(C) 1} =A_1(u_0), \\ \nonumber
	A_{(C) 2} =\displaystyle\frac{x^2}{u_0}A_0(u_0) - \displaystyle\frac{x^0}{u_0^2}A_2(u_0).
\end{gather*}

\subsection{The set $\{  p_1, \frac 12( p_0 + p_2) \} $ }
\label{subsect5}

Consider now the set  $\{  p_1, \frac 12( p_0 + p_2) \}$  (see  table  \ref{complete-sets-1}).
Although this set leads to a non-orthogonal separable coordinate system,
separation of variables in the Dirac equation (\ref{dirac1}) goes like in the case of the set
$\{p_0, p_1 \}$ considered in  subsection \ref{subsect1}.
Therefore, we omit here technical details and just give the main results.

According to  (\ref{one-to-one-1}), we find the symmetry operators $\{X_1, X_2 \}$ as
\begin{eqnarray}
\label{set5}
X_1=  \mathcal{P}_1 + \underset{1}{\varphi} + {\underset{1}{\varphi}}\,_\alpha \gamma^\alpha , \quad
X_2= \displaystyle\frac{1}{2}\big(\mathcal{P}_0+\mathcal{P}_2 \big) + \underset{2}{\varphi} + {\underset{2}{\varphi}}\,_\alpha \gamma^\alpha .
\end{eqnarray}
The Killing vectors (\ref{killing-1}) are ${\underset{1}{\xi}}\,^\nu=\delta^\nu_1$, ${\underset{2}{\xi}}\,^\nu=( \delta_0^\nu+\delta^\nu_2)/2$. From (\ref{de032}) we find ${\underset{1}{\varphi}}\,_\nu = {\underset{2}{\varphi}}\,_\nu=0$. The  vector $\underset{2}\xi^\nu$ is straightened out in a coordinate system $(u_0,u_1,u_2)$,
\begin{eqnarray}
\label{set5-1}
u_0= x^0 - x^2, \quad u_1=x^1, \quad u_2=x^0+ x^2.
\end{eqnarray}
In  coordinates (\ref{set5-1}), the  components of the metric tensor are
\begin{equation*}
(g_{\mu\nu})=
\left(
\begin{array}{ccc}
0 & 0 &  2 \\
0 & -1 & 0 \\
2 & 0  & 0 \\
\end{array}
\right), \quad
(g^{\mu\nu})=
\left(
\begin{array}{ccc}
0 & 0 &  \frac 12 \\
0 & -1 & 0 \\
\frac 12 & 0  & 0 \\
\end{array}
\right).
\end{equation*}
The Killing vectors take the Kronecker delta form $\underset{1}{\xi}^\nu=\delta^\nu_1$, $\underset{2}{\xi}^\nu=\delta^\nu_2$. From  (\ref{levi-chivita1}) and (\ref{spinconn3}) it follows that
 \begin{equation}
\label{set5-2-1}
\Gamma^\nu_{\mu\alpha}=\Gamma_\nu=0.
\end{equation}
From (\ref{set5-2-1}) and  Lemma \ref{lemma}, we find a potential permitting the set of symmetry operators (\ref{set5}) in coordinates (\ref{set5-1}) as $A_\nu=A_\nu(u_0)$, where $A_\nu(u_0)$ are arbitrary functions of $u_0$. In addition, we find that $\underset{1}\varphi=A_1$ and $\underset{2}\varphi=A_2$, and then
\begin{eqnarray}
	&& X_1=i\displaystyle\partial_{u_1}, \quad
	   X_2=i\displaystyle\partial_{u_2}.
	\label{set5-44}
\end{eqnarray}

Separable solutions to the Dirac equation (\ref{dirac1}) are obtained from  system  (\ref{set1-4-2}) with operators
(\ref{set5-44}) and have the form
\begin{gather*}
\psi=e^{-i\lambda_1 u_1-i\lambda_2 u_2}\tilde\psi(u_0),\quad
\tilde\psi(u_0)=
\left(
\begin{array}{c}
     \tilde\psi_1(u_0)\\
     \tilde\psi_2(u_0)\\
     \end{array}
\right),\\ \nonumber
\bigg{\{}
		\left(\hat\gamma^0-\hat\gamma^2\right) \big(i\displaystyle\frac{d}{d u_0}-A_0(u_0) \big)
		+ \big( \hat\gamma^0+\hat\gamma^2 \big)\big(\lambda_2-A_2(u_0) \big) +\\ \nonumber
		+\hat\gamma^1\big(\lambda_1-A_1(u_0) \big) -m \bigg{\}}\tilde\psi(u_0)=0.
\end{gather*}

\textbf{Catresian coordinates}.
The solution $\phi_C$ (\ref{set1-4-12})  and the solution $\psi$ from (\ref{set1-4-2}) in the curvilinear separable coordinates (\ref{set5-1}) are identical: $\phi_C=\psi$. The Cartesian components  $A_{(C)\alpha}$ of the potential are
\begin{eqnarray*}
A_{(C) 0} =A_0(u_0)+A_2(u_0), \quad A_{(C) 1} =A_1(u_0), \quad A_{(C) 2} =-A_0(u_0)+A_2(u_0).
\end{eqnarray*}

\subsection{The set $\{ \frac 12(p_0 + p_1),  l_{21}+l_{01}+a(p_0-p_2) \} $ }
\label{subsect6}

For the set  $\{ \frac 12(p_0 + p_1),  l_{21}+l_{01}+a(p_0-p_2) \}$ (see table \ref{complete-sets-1}) in accordance with (\ref{one-to-one-1}),
the corresponding set of  symmetry operators  $\{X_1, X_2 \}$ in Cartesian coordinates is written as
\begin{gather}\label{set6ab}
	X_1 =  \displaystyle\frac{1}{2}(\mathcal{P}_0+\mathcal{P}_2) + \underset{1}{\varphi} + {\underset{1}{\varphi}}\,_\alpha \gamma^\alpha,\quad
	X_2 = L_{21}+L_{01}+ a(\mathcal{P}_0-\mathcal{P}_2)+ \underset{2}{\varphi} + {\underset{2}{\varphi}}\,_\alpha \gamma^\alpha.
\end{gather}
The corresponding  Killing vectors (\ref{killing-1}) are
\begin{gather}\label{set6cd}
	\underset{1}\xi^\nu = \displaystyle\frac{1}{2}(\delta_0^\nu+\delta_2^\nu), \quad
	\underset{2}\xi^\nu = (x^1+a)\delta_0^\nu+(x^0-x^2)\delta_1^\nu+(x^1-a)\delta_2^\nu.
\end{gather}
The Killing vectors (\ref{set6cd}) can be straightened out in a coordinate system $(u_0,u_1,u_2)$:
\begin{eqnarray}
\label{set6-1}
&& u_0=\displaystyle\frac{1}{2}( x^0 - x^2)^2 -2 a x^1, \cr
&& u_1=x^0+x^2 +\displaystyle\frac{1}{a}(x^0-x^2)\left[\displaystyle\frac{(x^0-x^2)^2}{6a} -x^1\right], \\
&& u_2= \displaystyle\frac{1}{2a}( x^0 - x^2).\nonumber
\end{eqnarray}
The inverse coordinate transformation is
\begin{eqnarray}
\label{set6-1-A}
&& x^0=\displaystyle\frac{1}{2} u_1 - \displaystyle\frac{1}{2a}u_0 u_2 +\displaystyle\frac{1}{3} au_2^3 +a u_2, \cr
&& x^1= a u_2^2 -\displaystyle\frac{1}{2a}u_0, \\
&& x^2=\displaystyle\frac{1}{2} u_1-\displaystyle\frac{1}{2a} u_0 u_2 + \displaystyle\frac{1}{3} u_2^3 -a u_2.\nonumber
\end{eqnarray}
Here we assume that $a\neq 0$. In  coordinates (\ref{set6-1}), (\ref{set6-1-A})  the metric tensor takes the form
\begin{equation*}
(g^{\mu\nu})=
\left(
\begin{array}{ccc}
-4 a^2 & 0 &  0 \\
0 & \displaystyle\frac{2 u_0}{a^2} & \displaystyle\frac{1}{a} \vspace{1mm}\\
0 & \displaystyle\frac{1}{a}  & 0 \\
\end{array}
\right), \quad
(g_{\mu\nu})=
\left(
\begin{array}{ccc}
-\displaystyle\frac{1}{4 a^2} & 0 & 0 \\
0 & 0 & a \\
0 & a  & -2 u_0 \\
\end{array}
\right).
\end{equation*}
The Killing vectors (\ref{set6cd}) take the Kronecker delta form $\underset{1}{\xi}^\nu=\delta^\nu_1$, $\underset{2}{\xi}^\nu=\delta^\nu_2$. The the non-zero  Christoffel symbols are $\Gamma^0_{22}=-4 a^2$, $\Gamma^1_{20}=-1/a$. Choose  local frame fields $(e^\nu_a (x))$ (\ref{triada1}) as
\begin{eqnarray}
\label{set6-3}
 &&(e^\nu_0) =\left(\displaystyle \partial_0 u_\nu \right)=\left(2au_2, 1+\displaystyle\frac{1}{2 a^2}u_0+u_2^2,
 \displaystyle\frac{1}{2a}\right), \cr
 && (e^\nu_1) =\left(\displaystyle \partial_1 u_\nu \right)= (-2a, -2 u_2,0), \cr
 && (e^\nu_2) =\left(\displaystyle\partial_2 u_\nu\right)= \left(-2au_2, 1-\displaystyle\frac{1}{2 a^2}u_0 -u_2^2,-\displaystyle\frac{1}{2a} \right).
     \end{eqnarray}
From (\ref{set6-1}) --- (\ref{set6-3}) and  Lemma \ref{lemma}, we find a potential permitting the set of symmetry operators (\ref{set6ab}) in coordinates (\ref{set6-1}), (\ref{set6-1-A}) as $A_\nu=A_\nu(u_0)$, where $A_\nu(u_0)$ are arbitrary functions of $u_0$. In addition, we find that $\underset{1}\varphi=A_1$ and $\underset{2}\varphi=A_2$, and from (\ref{de032}) it follows that $\underset{1}{\varphi}\,_\nu=0$, $\underset{2}{\varphi}\,_\nu=-as\delta_{2\nu}$. Then the symmetry operators (\ref{set6ab}) can be written as
\begin{eqnarray}
	&& X_1=i\displaystyle\partial_{u_1}, \quad
	   X_2=i\displaystyle\partial_{u_2}-a s \gamma^2,
\label{set6-15}
 \end{eqnarray}
and the Dirac equation (\ref{dirac1}) in coordinates (\ref{set6-1}), (\ref{set6-1-A}) takes the form (\ref{set1-4-2-A}), where it is necessary to put
\begin{gather*}
	\gamma^0 = 2 a u_2(\hat\gamma^0-\hat\gamma^2)-2a\hat\gamma^1, \quad
	\gamma^2 = \displaystyle\frac{1}{2a}(\hat\gamma^0-\hat\gamma^2),\\ \nonumber
  	\gamma^1 = \hat\gamma^0+\hat\gamma^2+\big(  \displaystyle\frac{1}{2a^2}u_0+u_2^2 \big)(\hat\gamma^0-\hat\gamma^2)-2u_2\hat\gamma^1.
\end{gather*}

Separable solutions to the Dirac equation (\ref{dirac1}) in  coordinates  (\ref{set6-1}), (\ref{set6-1-A}) are found from system (\ref{set1-4-2}), (\ref{set1-4-2-A}) with operators (\ref{set6-15}) and can be written as
\begin{gather*}
	\psi = \exp\left(-i\lambda_1 u_1-i\lambda_2 u_2-i\frac{s}{2}(\hat\gamma^0-\hat\gamma^2) u_2 \right)\tilde\psi(u_0),\quad
	\tilde\psi(u_0)=
		\left(
			\begin{array}{c}
			     \tilde\psi_1(u_0)\\
			     \tilde\psi_2(u_0)\\
		        \end{array}
		\right),\\
	\bigg\{ -2a\hat\gamma^1 \left( i\displaystyle\frac{\partial}{\partial u_0}-A_0(u_0)\right)+
	\displaystyle\frac{1}{2a}\big(\hat\gamma^0-\hat\gamma^2\big)\big(\lambda_2-A_2(u_0)\big)+ \\ \nonumber
 	+ \Big(\hat\gamma^0+\hat\gamma^2 +\displaystyle\frac{u_0 }{2a^2}\big(\hat\gamma^0-\hat\gamma^2\big)\Big) \big(\lambda_1-A_1(u_0)\big)-m \bigg\}\tilde\psi(u_0)=0.
\end{gather*}

\textbf{Catresian coordinates}.
The solution $\phi_C$ (\ref{set1-4-12})  and the solution $\psi$ from (\ref{set1-4-2-A})  in the curvilinear separable coordinates (\ref{set6-1}), (\ref{set6-1-A})  are identical: $\phi_C=\psi$. The Cartesian components  $A_{(C)\alpha}$ of the potential are
\begin{eqnarray*}
&&A_{(C) 0} =(x^0-x^2)A_0(u_0)+\big(1-\displaystyle\frac{1}{a}x^1+\displaystyle\frac{1}{2a^2}(x^0-x^2)^2 \big) A_1(u_0)+ \displaystyle\frac{1}{2a}A_2 (u_0), \cr
&& A_{(C) 1} =-2a A_0(u_0)- \displaystyle\frac{1}{a}(x^0-x^2) A_1(u_0), \\
&& A_{(C) 2} =-(x^0-x^2)A_0(u_0)+\big(1+\displaystyle\frac{1}{a}x^1-\displaystyle\frac{1}{2a^2}(x^0-x^2)^2 \big) A_1(u_0)- \displaystyle\frac{1}{2a}A_2(u_0). \nonumber
 \end{eqnarray*}

\subsection{The set $\{ \frac 12(p_0 + p_1),  l_{21}+l_{01} \} $ }
\label{subsect7}

Now we consider  the set  $\{ \frac 12(p_0 + p_1),  l_{21}+l_{01}+a(p_0-p_2) \}$  (see table  \ref{complete-sets-1})
when $a=0$. In accordance with (\ref{one-to-one-1}), the corresponding set of  symmetry operators  $\{X_1, X_2 \}$ in Cartesian coordinates can be written as
\begin{eqnarray}\label{set7ab}
	X_1 =  \displaystyle\frac{1}{2}(\mathcal{P}_0+\mathcal{P}_2) + \underset{1}{\varphi} + {\underset{1}{\varphi}}\,_\alpha \gamma^\alpha,\quad
	X_2 = L_{21}+L_{01}+ \underset{2}{\varphi} + {\underset{2}{\varphi}}\,_\alpha \gamma^\alpha.
\end{eqnarray}
The corresponding  Killing vectors (\ref{killing-1}) are
\begin{eqnarray}\label{set7cd}
	\underset{1}\xi^\nu =  \displaystyle\frac{1}{2}(\delta_0^\nu+\delta_2^\nu), \quad
	\underset{2}\xi^\nu= x^1\delta_0^\nu+(x^0-x^2)\delta_1^\nu+x^1\delta_2^\nu.
\end{eqnarray}
The Killing vectors (\ref{set7cd}) can be straightened out in a coordinate system $(u_0,u_1,u_2)$:
\begin{eqnarray}
\label{set7-1}
&& u_0= x^0 - x^2, \quad  u_1=x^0+x^2 -\displaystyle\frac{(x^1)^2}{x^0-x^1}, \quad u_2= \displaystyle\frac{x^1}{x^0-x^2}.
\end{eqnarray}
 The inverse coordinate transformation is
\begin{equation}
	\label{set7-1-A}
	x^0=\displaystyle\frac{1}{2} (u_0+u_1+u_0 u_2^2),\quad
	x^1 = u_0 u_2, \quad
	x^2 = x^0 = \displaystyle\frac{1}{2} (-u_0+u_1+u_0 u_2^2).
\end{equation}
In  coordinates (\ref{set7-1}), the metric tensor is
\begin{equation*}
(g^{\mu\nu})=
\left(
\begin{array}{ccc}
0 & 2 &  0 \\
2 & 0 & 0 \vspace{1mm}\\
0 & 0 & -\displaystyle\frac{1}{u_0^2} \\
\end{array}
\right), \quad
(g_{\mu\nu})=
\left(
\begin{array}{ccc}
0 & \displaystyle\frac{1}{2} & 0 \\
\displaystyle\frac{1}{2} & 0 & 0 \\
0 & 0  & - u_0^2 \\
\end{array}
\right),\quad
g=\mbox{det}(g_{\nu\mu})=\displaystyle\frac{1}{4}u_0^2.
\end{equation*}
The Killing vectors (\ref{set7cd}) take the Kronecker delta form $\underset{1}{\xi}^\nu=\delta^\nu_1$, $\underset{2}{\xi}^\nu=\delta^\nu_2$. The the non-zero  Christoffel symbols are $\Gamma^1_{22}=2 u_0$, $\Gamma^2_{20}=1/u_0$. Choose local frame fields $(e^\nu_a (x))$
(\ref{triada1}) as
\begin{eqnarray*}
 &&(e^\nu_0) =\left(\displaystyle\partial_0 u_\nu \right)=\left(1, 1+u_2^2, - \displaystyle\frac{u_2}{u_0}\right), \cr
 && (e^\nu_1) =\left(\displaystyle\partial_1 u_\nu \right)= (0, -2u_2,\displaystyle\frac{1}{u_0}), \cr
 && (e^\nu_2) =\left(\displaystyle\partial_2 u_\nu\right)= \left(-1, 1-u_2^2,\displaystyle\frac{u_2}{u_0} \right).
\end{eqnarray*}

From (\ref{de032}) and  Lemma \ref{lemma}, we find a potential permitting the set of symmetry operators (\ref{set7ab}) in coordinates (\ref{set7-1}), (\ref{set7-1-A}) as $A_\nu=A_\nu(u_0)$, where $A_\nu(u_0)$ are arbitrary functions. In addition, we have  $\underset{1}\varphi=A_1$,  $\underset{2}\varphi=A_2$, and from (\ref{de032}) it follows that $\underset{1}{\varphi}\,_\nu=0$, $\underset{2}{\varphi}\,_\nu=-(s/2)\delta_{0\nu}$.

Then the symmetry operators  (\ref{set7ab}) can be written as
\begin{eqnarray}
	&& X_1=i\displaystyle\partial_{u_1}, \quad
	   X_2=i\displaystyle\partial_{u_2}-\displaystyle\frac{s}{2}s\gamma^0,
\label{set7-15}
 \end{eqnarray}
  and the Dirac equation (\ref{dirac1}) in coordinates (\ref{set7-1}), (\ref{set7-1-A}) takes
the form (\ref{set1-4-2-A}), where it is necessary to put
\begin{eqnarray*}
\gamma^0=\hat\gamma^0-\hat\gamma^2, \quad \gamma^1=\hat\gamma^0+\hat\gamma^2-2u_2 \gamma^1 +u_2^2 (\hat\gamma^0-\hat\gamma^2), \quad
\gamma^2=-\displaystyle\frac{u_2}{u_0}(\hat\gamma^0-\hat\gamma^2) +\displaystyle\frac{1}{u_0}\hat\gamma^1.
\end{eqnarray*}

Separable solutions to the Dirac equation (\ref{dirac1}) in  coordinates  (\ref{set7-1}), (\ref{set7-1-A})
 are found from system (\ref{set1-4-2}),  (\ref{set1-4-2-A}) with operators (\ref{set7-15}) and
 can be written as
\begin{gather*}
	\psi = \exp\left(-i\lambda_1 u_1-i\lambda_2 u_2-i\frac{s}{2}(\hat\gamma^0-\hat\gamma^2)u_2 \right)\tilde\psi(u_0),\quad
	\tilde\psi(u_0)=
		\left(
			\begin{array}{c}
			     \tilde\psi_1(u_0)\\
			     \tilde\psi_2(u_0)\\
			\end{array}
		\right),\\
	\bigg\{\big(\hat\gamma^0-\hat\gamma^2\big) \left( i\displaystyle\frac{d}{d u_0}+\displaystyle\frac{i}{2 u_0}-A_0(u_0)\right)+ \big(\hat\gamma^0 + \hat\gamma^2\big)\big(\lambda_1-A_1(u_0)\big)+ \\ \nonumber
	+ \displaystyle\frac{1}{u_0}\hat\gamma^1\big(\lambda_2-A_2(u_0)\big)- m\bigg\}\tilde\psi(u_0)=0.
\end{gather*}

\textbf{Catresian coordinates}.
The solution $\phi_C$ (\ref{set1-4-12})  and the solution $\psi$ from (\ref{set1-4-2-A}) in the curvilinear separable coordinates (\ref{set7-1}), (\ref{set7-1-A})  are identical: $\phi_C=\psi$. The Cartesian components  $A_{(C)\alpha}$ of the potential are
\begin{eqnarray*}
&&A_{(C) 0} =A_0(u_0)+\Big(1+\displaystyle\frac{(x^1)^2}{(x^0-x^2)^2}\Big)A_1(u_0)-\displaystyle\frac{x_1}{(x^0-x^2)^2} A_2(u_0), \cr
&& A_{(C) 1} =- A_0(u_0)+ \Big(1-\displaystyle\frac{(x^1)^2}{(x^0-x^2)^2}\Big)A_1(u_0)+\displaystyle\frac{x_1}{(x^0-x^2)^2} A_2(u_0), \cr
&& A_{(C) 2} = -\displaystyle\frac{2 x^1}{x^0-x^2}A_1(u_0)+\displaystyle\frac{1}{x^0-x^2}A_2(u_0).
 \end{eqnarray*}

\section{Conclusion}
\label{concrem}

We have considered the problem of separation of variables in the  $(2+1)$-dimensional Dirac equation with external electromagnetic potential in Minkowski spacetime with the use of first-order differential symmetry operators with the derivatives having matrix coefficients.

To this end we explored the properties of  symmetry operators in   $(2+1)$-dimensional pseudo-Riemannian space in terms of the determining equations for the 
operator.

The symmetry operators under consideration which are analytical (holomorphic) with respect to the mass parameter $m$ entering into the Dirac equation, are shown to have different properties.
 In a massive case ($m\neq 0$), a symmetry operator commutes with the operator of the Dirac equation and has scalar (non-matrix) coefficients of the derivatives. The complete set of symmetry operators provides separation of variables in the Dirac equation in ($2+1$)-dimensional Minkowski spacetime like in the ($3+1$)-dimensional case  \cite{bagr_git_book2014}. The compete sets have been classified and for each of them,  separation of variables has been carried out.

In the curved pseudo-Riemann $(2+1)$-dimensional space, as well as in the $(3+1)$-dimensional  space, the problem of separation of variables in the Dirac equation  has specific features as compared to that in the plane Minkowski space and it is to be a subject of separate study.

From the point of view of developing a theory of separation of variables  in  the Dirac equation, the $(2+1)$-dimensional case is especially attractive as it is mathematically simpler but includes all basic elements of the theory. Unlike $(3+1)$-dimensional space, in a massive case, the symmetry operators for the $(2+1)$-dimensional Dirac equation are presented in terms of the Killing vectors, and  the spin operators with matrix coefficients of  the derivatives can be removed from the symmetry operators without loss of generality.

In a massless case ($m=0$), the commutator (\ref{determ-eq1}) of the symmetry operator with the $(2+1)$-dimensional Dirac operator $H$ in (\ref{dirac1}) is proportional to the Dirac operator with a Lagrange multiplier having the form of a scalar function. The symmetry operator is presented in terms of a conformal Killing vector field $\xi^\mu$. Therefore, the set of symmetry operators is wider in this case and the problem of separation of variables calls for a particular research.

In summary, we note that to use symmetry operators of both  the first and the second order for separation of variables in the Dirac equation, one can turn to the squared Dirac equation \cite{bagr92}. However, this problem is beyond the scope of the present work.

\section*{Acknowledgements}
The work was supported in part by Tomsk State University under the International Competitiveness
Improvement Program
and by Tomsk Polytechnic University under the International Competitiveness
Improvement Program.


\end{document}